\documentclass[11pt,a4paper]{amsart}
\usepackage[foot]{amsaddr}
\usepackage{ifxetex}
\ifxetex
  \usepackage[no-math]{fontspec}
\else
\fi
\usepackage{amsmath}
\usepackage{amsfonts}
\usepackage{amssymb}
\usepackage{amsthm}
\usepackage{fullpage}
\usepackage{tablefootnote}
\usepackage{microtype}
\ifxetex
  \usepackage[libertine]{newtxmath}
\else
  \usepackage{newtxmath}
\fi
\usepackage[tt=false]{libertine}
\usepackage{caption}
\usepackage{tcolorbox}
\usepackage{bbm}
\usepackage{hyperref, color}
\hypersetup{colorlinks=true,citecolor=blue, linkcolor=blue, urlcolor=blue}
\usepackage[linesnumbered,boxed,ruled,vlined]{algorithm2e}
\usepackage{bm}
\usepackage{bbm}
\usepackage[numbers]{natbib}
\usepackage{xcolor}
\usepackage{enumerate}
\usepackage{enumitem}
\usepackage{ragged2e}
\usepackage{tabularx}
\usepackage{array}
\usepackage{longtable}
\usepackage{hhline}
\usepackage{multirow}
\newcolumntype{L}[1]{>{\=Raggedright\arraybackslash}p{#1}}
\newcolumntype{C}[1]{>{\centering\arraybackslash}m{#1}}
\newcolumntype{R}[1]{>{\=Raggedleft\arraybackslash}p{#1}}

\usepackage{makecell}
\usepackage{aliascnt}
\usepackage{cleveref}
\usepackage{footnote}
\usepackage{float}
\makesavenoteenv{tabular}

\renewcommand{\epsilon}{\varepsilon}

\newtheorem{theorem}{Theorem}[section]
\newcommand{\newsharedtheorem}[2]{%
  \newaliascnt{#1}{theorem}%
  \newtheorem{#1}[#1]{#2}%
  \aliascntresetthe{#1}%
}
\newsharedtheorem{observation}{Observation}
\newsharedtheorem{claim}{Claim}
\newtheorem*{claim*}{Claim}
\newsharedtheorem{condition}{Condition}
\newsharedtheorem{example}{Example}
\newsharedtheorem{fact}{Fact}
\newsharedtheorem{lemma}{Lemma}
\newsharedtheorem{proposition}{Proposition}
\newsharedtheorem{conjecture}{Conjecture}
\newsharedtheorem{corollary}{Corollary}
\theoremstyle{definition}

\newsharedtheorem{definition}{Definition}
\newsharedtheorem{remark}{Remark}
\newtheorem*{remark*}{Remark}

\crefname{theorem}{theorem}{theorems}
\Crefname{theorem}{Theorem}{Theorems}
\crefname{observation}{observation}{observations}
\Crefname{observation}{Observation}{Observations}
\crefname{claim}{claim}{claims}
\Crefname{claim}{Claim}{Claims}
\crefname{condition}{condition}{conditions}
\Crefname{condition}{Condition}{Conditions}
\crefname{Condition}{condition}{conditions}
\Crefname{Condition}{Condition}{Conditions}
\crefname{example}{example}{examples}
\Crefname{example}{Example}{Examples}
\crefname{fact}{fact}{facts}
\Crefname{fact}{Fact}{Facts}
\crefname{lemma}{lemma}{lemmas}
\Crefname{lemma}{Lemma}{Lemmas}
\crefname{proposition}{proposition}{propositions}
\Crefname{proposition}{Proposition}{Propositions}
\crefname{conjecture}{conjecture}{conjectures}
\Crefname{conjecture}{Conjecture}{Conjectures}
\crefname{corollary}{corollary}{corollaries}
\Crefname{corollary}{Corollary}{Corollaries}
\crefname{definition}{definition}{definitions}
\Crefname{definition}{Definition}{Definitions}
\crefname{remark}{remark}{remarks}
\Crefname{remark}{Remark}{Remarks}
\crefname{assumption}{assumption}{assumptions}
\Crefname{assumption}{Assumption}{Assumptions}

\def\\+Ex{\mathop{\mathbf{\+E}}\nolimits}
\renewcommand{\Pr}[2][]{ \ifthenelse{\isempty{#1}}
  {\mathop{\mathbf{Pr}}\left[#2\right]} {\mathop{\mathbf{Pr}}_{#1}\left[#2\right]} }

\newcommand{\abs}[1]{\left\vert#1\right\vert}
\newcommand{\set}[1]{\left\{#1\right\}}
\newcommand{\tuple}[1]{\left(#1\right)}

\newcommand{\tp}{\tuple}

\newcommand{\inner}[2]{\left\langle #1,#2\right\rangle}

\newcommand{\defeq}{\triangleq}

\newcommand{\DTV}[2]{d_{\mathrm{TV}}\left({#1},{#2}\right)}

\newcommand{\edge}{\mathrm{edge}}
\newcommand{\Cov}{\mathrm{Cov}}

\def\*#1{\boldsymbol{#1}}
\def\+#1{\mathcal{#1}}
\def\-#1{\mathrm{#1}}
\def\=#1{\mathbb{#1}}

\DeclareMathOperator{\oPr}{\mathbf{Pr}}
\renewcommand{\Pr}[2][]
{\ifthenelse{\isempty{#1}}
  {\oPr\left[#2\right]}
  {\oPr_{#1}\left[#2\right]}
} 

\DeclareMathOperator{\oE}{\mathbb{E}}
\newcommand{\E}[2][]
{\ifthenelse{\isempty{#1}}
  {\oE\left[#2\right]}
  {\oE_{#1}\left[#2\right]}
} 

\DeclareMathOperator{\oVar}{\mathrm{Var}}
\newcommand{\Var}[2][]
{\ifthenelse{\isempty{#1}}
  {\oVar\left[#2\right]}
  {\oVar_{#1}\left[#2\right]}
}
\def\oEnt{\mathrm{Ent}}
\NewDocumentCommand{\Ent}{ O{} O{} m }{
  \ifthenelse{\isempty{#1}} {
    \ifthenelse{\isempty{#2}} {
      \oEnt\left[#3\right]
    } {
      \oEnt^{#2}\left[#3\right]
    }
  } {
    \ifthenelse{\isempty{#2}} {
      \oEnt_{#1}\left[#3\right]
    } {
      \oEnt_{#1}^{#2}\left[#3\right]
    }
  }
}

\usepackage[textsize=tiny]{todonotes}

\usepackage{xifthen}

\newcommand{\dist}{\mathrm{dist}}
\newcommand{\one}[1]{\mathbbm{1}_{\{#1\}}}
\makeatletter
\def\prob#1#2#3{\goodbreak\begin{list}{}{\labelwidth\z@ \itemindent-\leftmargin
                        \itemsep\z@  \topsep6\p@\@plus6\p@
                        \let\makelabel\descriptionlabel}
                \item[\it Name]#1
               \item[\it Instance]                #2
                \item[\it Output]#3
                \end{list}}
\makeatother

\newbool{doubleblind}
\setbool{doubleblind}{false}

\title{Approximate Colorwise Tensorization of Entropy and\\ Optimal Mixing of the Wang--Swendsen--Koteck\'{y} Dynamics}

\ifdoubleblind
\author{Author(s)}
\else
\author{Chunyang Wang, Yuichi Yoshida, Zihan Zhang}
\address[Chunyang Wang, Yuichi Yoshida, Zihan Zhang]{National Institute of Informatics, Tokyo, Japan. \textnormal{Email: \texttt{$\{$c\_wang,yyoshida,zihan$\}$@nii.ac.jp}}}

\fi

\begin{document}

\begin{abstract}
We study the mixing time of Wang--Swendsen--Koteck\'{y} (WSK) dynamics for uniformly sampling proper $q$-colorings. The WSK dynamics is widely used in statistical physics for sampling from the antiferromagnetic Potts model and can be considered a global counterpart of the flip dynamics, which currently yields the state-of-the-art bounds for sampling colorings in general graphs (Carlson and Vigoda, SODA 2025). However, despite its importance, the tools for analyzing such dynamics remain limited.

We develop new tools that enable us to analyze the mixing time of the WSK dynamics through the lens of relative entropy contraction. We introduce new criteria for multi-spin distributions: approximate colorwise tensorization of entropy (ACTE) and approximate colorwise subadditivity of entropy (ACSE). These criteria provide a colorwise counterpart to standard vertex-wise entropy factorization principles, and expose a form of color symmetry beyond coordinate-wise analyses. We also develop new inductive approaches for establishing such criteria on specific types of graphs, which can be viewed as local-to-global arguments for proving high-dimensional functional inequalities in a graph-theoretic sense.

As concrete applications, we establish an optimal $O_q(\log n)$ mixing time for the WSK dynamics on chordal and outerplanar graphs, down to the optimal number of colors. Because trees and line graphs of trees are chordal, the result covers both vertex and edge colorings of trees. Our results work in a regime that bypasses the irreducibility threshold for Glauber dynamics while also improving the best known mixing time bounds (Carlson, Chen, Feng and Vigoda, SODA 2025). 



\end{abstract}	

\maketitle

\setcounter{tocdepth}{1}
\tableofcontents

\section{Introduction}

Sampling proper $q$-colorings remains one of the most fundamental yet notoriously difficult problems in the field of approximate counting and sampling. A longstanding open conjecture postulates the existence of an efficient algorithm for approximately counting and sampling proper $q$-colorings of graphs with maximum degree $\Delta$ whenever $q \geq \Delta+1$. Despite a storied history, we are still far from proving this conjecture. To date, the best threshold for efficient approximate counting and sampling is $q>1.809\Delta$, recently established by Carlson and Vigoda~\cite{carlson2025flip}. This result represents the frontier of a line of work~\cite{vigoda1999improved,chen2019improved,carlson2025flip} analyzing a Markov chain known as \emph{flip dynamics}, wherein a two-color Kempe component is chosen and potentially flipped at each step. 

An alternative Markov chain for sampling proper $q$-colorings is the \emph{Wang--Swendsen--Koteck\'{y} (WSK) dynamics}~\cite{wang1989antiferro,wang1990three}, which is also based on Kempe moves and can be viewed as a global counterpart of the flip dynamics. The WSK dynamics originated in statistical physics as a cluster algorithm for the antiferromagnetic $q$-state Potts model. Although it is less studied from a theoretical perspective, it is widely used in practice and shown to be empirically effective on specific instances. Numerical studies~\cite{ferreira1999antiferro,sokal2001personal} have suggested that the WSK dynamics might exhibit constant-time mixing for the critical~\cite{alon2021mixing} $q=3$ regime on $n\times n$ periodic square lattices with even $n$.

The WSK dynamics for sampling proper $q$-colorings is formally defined as follows:

 \begin{definition}[Wang--Swendsen--Koteck\'{y} dynamics]\label{def:WSK-dynamics}
Let $G=(V,E)$ be a graph and let $q\geq 2$ be an integer such that $G$ admits at least one proper $q$-coloring. Given a current proper $q$-coloring $\sigma\in [q]^V$, a single step of the Wang--Swendsen--Koteck\'{y} (WSK) dynamics proceeds as follows:

First, choose an unordered pair of distinct colors $\{a,b\}\subseteq[q]$ uniformly at random. Let
\[
    W=W_{a,b}(\sigma)\defeq\{v\in V:\sigma_v\in\{a,b\}\}
\]
denote the set of vertices colored either $a$ or $b$. Sample a new coloring $\sigma'$ uniformly from the proper $q$-colorings satisfying
\[\sigma'_{V\setminus W}=\sigma_{V\setminus W},\quad \sigma'_v\in\{a,b\} \text{ for all } v\in W.\]

Equivalently, for each connected component of the induced subgraph $G[W]$, independently choose to either retain its current coloring or swap the colors $a$ and $b$, each with probability $1/2$.
\end{definition}

\medskip
\noindent
\textbf{Challenges for analyzing the mixing time of the WSK dynamics/flip dynamics.}\,\, Despite the theoretical significance of flip dynamics and the practical relevance of the WSK dynamics, the mathematical toolkit for analyzing the mixing time of such Kempe-swap Markov chains remains much less developed than that for Glauber dynamics. Historically, the mixing time analysis of flip dynamics has relied predominantly on path coupling~\cite{bubley97pathcoupling,vigoda1999improved,chen2019improved,carlson2025flip}, a technique introduced nearly thirty years ago. 
The same is true for the WSK dynamics. Although the WSK dynamics has long been used in statistical physics as a natural cluster algorithm for the zero-temperature antiferromagnetic Potts model, rigorous work on the WSK dynamics has focused mainly on its (non-)ergodicity~\cite{mohar2009kempe,mohar2010nonergodicity,salas2022ergodicity}, and torpid-mixing lower bounds for particular graph families~\cite{luczak2005torpid}. To the best of our knowledge, there is currently no general rapid-mixing upper-bound theory for the WSK dynamics on proper $q$-colorings.

This stands in sharp contrast with the recent breakthroughs in the analysis of Markov chain mixing times. Over the past few years, a powerful new class of techniques based on high-dimensional expanders and spectral independence has led to several major advances for the analysis of mixing times of Glauber dynamics~\cite{anari2019logconcave, alev2020improved, anari2020spectral,chen2020rapid}. In particular, for antiferromagnetic two-spin systems, rapid mixing of
the Glauber dynamics up to the uniqueness threshold has been established in a series of works~\cite{anari2019logconcave, alev2020improved, anari2020spectral,chen2020rapid, anari2021logconcave,chen2021optimal,chen2021rapid,abdolazimi2022matrix,anari2022entropic,blanca2022mixing,CE22,CFYZ22}. However, while this framework has achieved optimal mixing bounds for antiferromagnetic two-spin systems, applying it to multi-spin systems like $q$-coloring has proven challenging without additional structural assumptions. Therefore, both directions beg the following question:

\begin{center}
    \emph{Can we leverage recent advances in Markov chain mixing times\\ to develop new tools for analyzing WSK and flip dynamics?}
\end{center}

\subsection{Our Contributions}
In this paper, we answer the aforementioned question positively by introducing a new framework for analyzing the mixing time of the WSK dynamics. We successfully combine the analysis of the WSK dynamics with recent advances in the analysis of Markov chain mixing times through functional inequalities. Our contribution can be divided into two parts.

Our first contribution is new entropy criteria for the WSK dynamics. We introduce \emph{approximate colorwise tensorization of entropy} (ACTE) and \emph{approximate colorwise subadditivity of entropy} (ACSE), the colorwise analogues of approximate tensorization of entropy and approximate subadditivity of entropy. In particular, we show that ACTE directly lower-bounds the modified log-Sobolev constant of the WSK dynamics, and ACSE over induced subgraphs implies ACTE.

\subsubsection{New Criteria for Rapid Mixing of the WSK dynamics}



Modern analytic techniques for bounding the mixing time of Markov chains establish decay to equilibrium by means of functional inequalities such as Poincar\'{e} or log-Sobolev inequalities, which correspond to decay
of variance and relative entropy respectively. In particular, modified log-Sobolev inequalities are often used to prove optimal mixing time bounds. While there have been
several notable recent results bounding the modified log-Sobolev
constant for various Markov chains~\cite{CGM19,blanca2021entropy,hermon2023modified}, we are not aware of any prior techniques or results addressing the modified log-Sobolev constant for the WSK dynamics. 

In this paper, we establish new criteria for multi-spin distributions for bounding the modified log-Sobolev constant and therefore the mixing time of the WSK dynamics. 

For a configuration $\sigma\in [q]^V$ and a color $c\in [q]$, we use $I_c(\sigma)=\{v\mid \sigma_v=c\}$ to denote the set of vertices colored $c$ in $\sigma$. For a set of colors $S\subseteq [q]$, we denote $I_S(\sigma)\defeq \tp{I_c(\sigma)}_{c\in S}$. We first introduce the notion of \emph{approximate colorwise tensorization of entropy} in \Cref{def:color-entropy-tensorization}. Note that the choice of all but two color classes was to align with the updates of the WSK dynamics.

\begin{definition}[Approximate colorwise tensorization of entropy (ACTE)]\label{def:color-entropy-tensorization}
 Let $q\geq 3$ and $\mu$ be a distribution over $\Omega\subseteq [q]^V$. We say that $\mu$ satisfies \emph{approximate colorwise tensorization of entropy (ACTE)} with constant $C>0$ if for all $f:\Omega\to \=R^{\geq 0}$ we have that
\[
\Ent[\mu]{f}\leq C\sum\limits_{\{a, b\}\in \binom{[q]}{2}}\E[\mu]{\Ent[\mu]{f\mid I_{[q]\setminus \{a,b\}}}},
\]
where $\E[\mu]{f}=\sum\limits_{\sigma\in \Omega}\mu(\sigma)f(\sigma)$ and $\Ent[\mu]{f}=\E[\mu]{f\log f}-\E[\mu]{f}\log \E[\mu]{f}$. 
\end{definition}

Note that the condition in \Cref{def:color-entropy-tensorization} is defined with respect to the Gibbs distribution instead of the WSK dynamics. We remark that product distributions satisfy $1$-approximate colorwise tensorization of entropy. 

We show that approximate colorwise tensorization of entropy in fact directly implies the rapid mixing of the WSK dynamics.

\begin{theorem}[ACTE implies rapid mixing]\label{thm:tensorization-of-entropy-implies-mixing}
    Let $q\geq 3$ and let $G=(V,E)$ be a $q$-colorable graph with $|V|=n$. Suppose that the uniform distribution over all proper $q$-colorings of $G$ satisfies approximate colorwise tensorization of entropy with constant $C$.
    Then it holds for the WSK dynamics with respect to $G$ and $q$ that the modified log-Sobolev constant is at least $2/(Cq(q-1))$ and the mixing time is bounded by \[T_{\mathrm{mix}}=O(Cq^2(\log n+\log \log q)).\]
\end{theorem}

We further introduce another notion, called \emph{approximate colorwise subadditivity of entropy}, given in \Cref{def:color-entropy-subadditivity}. 

\begin{definition}[Approximate colorwise subadditivity of entropy (ACSE)]\label{def:color-entropy-subadditivity}
 Let $q\geq 3$ and $\mu$ be a distribution over $\Omega\subseteq [q]^V$. We say that $\mu$ satisfies \emph{approximate colorwise subadditivity of entropy (ACSE)} with constant $C>0$ if for all $f:\Omega\to \=R^{\geq 0}$ we have that
\[
\Ent[\mu]{f}\leq C\sum\limits_{c\in [q]} \E[\mu]{\Ent[\mu]{f\mid I_c}}.
\]
\end{definition}

We show that approximate colorwise subadditivity of entropy over all induced subgraphs is in fact sufficient to imply approximate colorwise tensorization of entropy, and hence rapid mixing of the WSK dynamics combining with \Cref{thm:tensorization-of-entropy-implies-mixing}.

\begin{theorem}[ACSE over induced subgraphs implies ACTE]\label{thm:subadditivity-implies-tensorization}
    Let $q\geq 3$ and let $G=(V,E)$ be a $q$-colorable graph. Suppose that for every $3\leq i\leq q$, there exists a constant $C(i)$ such that:
    \begin{itemize}
        \item for every induced subgraph $G'$ of $G$ that admits at least one proper $i$-coloring, the uniform distribution $\mu_{G',i}$ over all proper $i$-colorings of $G'$ satisfies approximate colorwise subadditivity of entropy with constant $C(i)$.
    \end{itemize}
    Then the uniform distribution $\mu_{G,q}$ over all proper $q$-colorings of $G$ satisfies approximate colorwise tensorization of entropy with constant
    \[
    \frac{2}{q(q-1)}\prod\limits_{i=3}^{q}  \tp{i C(i)}.
    \]
\end{theorem}

The form of \Cref{thm:subadditivity-implies-tensorization} is reminiscent of local-to-global arguments found in the high-dimensional expander literature~\cite{alev2020improved}. In fact, if revealing a single color class is viewed as a ``local'' operation, then \Cref{thm:subadditivity-implies-tensorization} can be understood as a ``colorwise local-to-global'' argument. A key distinction of \Cref{thm:subadditivity-implies-tensorization} is that, unlike traditional local-to-global theorems which require the local property to hold under all arbitrary pinnings of the original distribution, we only require the local property to hold for all induced subgraphs. This relaxation is possible because we exploit the color symmetry preserved under the specific type of conditionings outlined in \Cref{def:color-entropy-tensorization}, a structural advantage that also enables our subsequent applications.

\begin{remark}[Color symmetry and extensions to list colorings]\label{rem:color-symmetry}
    Previous methods for studying the uniform sampling of $q$-colorings do not depend essentially on the symmetry between colors. This includes path coupling (both for Glauber dynamics and flip dynamics)~\cite{bubley97pathcoupling,jerrum1995simple,salas1997absense,vigoda1999improved,chen2019improved,carlson2025flip}, correlation decay method~\cite{goldberg2005strong,chen2023strong}, spectral independence method~\cite{feng2021rapid,chen2021rapid,abdolazimi2022matrix,wang2024samplingproper} and zero-freeness/interpolation method~\cite{bencs2021zero,liu2025correlation}. As a result, these methods naturally extend to \emph{list colorings} where each vertex can have different color lists, and only the size of the color list of each vertex is required. In contrast, our framework strictly relies on the symmetry between colors, and therefore does not appear to generalize to list colorings in its current form.

    We remark that this limitation may be meaningful rather than merely technical.  List coloring
can be substantially more restrictive than ordinary coloring: for example, there
are graph families whose list chromatic number is much larger than their
chromatic number~\cite{erdos1980choosability}.  Consequently, one should not
assume without proof that the optimal sampling threshold for worst-case
$q$-list-colorings coincides with the corresponding threshold for ordinary
$q$-colorings on the same graph class.  In this sense, a technique that genuinely
uses the global color symmetry of the ordinary coloring model may potentially
bypass obstructions that are intrinsic to the adversarial list-coloring setting.
\end{remark}

\begin{remark}[Approximate tensorization and subadditivity of entropy]
    The naming of the criteria in \Cref{def:color-entropy-tensorization} and \Cref{def:color-entropy-subadditivity} aims to align with the historical definition of approximate \emph{tensorization} and \emph{subadditivity} of entropy, which are notions useful for bounding the modified log-Sobolev constant for Glauber dynamics. 
    
    Approximate tensorization of entropy was shown to directly imply bounds on the modified log-Sobolev constant and the mixing time of the Glauber dynamics~\cite{caputo2015approximate}, and played a central role in
classical results proving that strong spatial mixing implies optimal mixing time
of the Glauber dynamics~\cite{stroock1992logarithmic,martinelli1994approach,cesi2001quasi}. Formally, a distribution $\mu$ over $\Omega\subseteq [q]^V$ is said to satisfy approximate tensorization of entropy with constant $C$ if for all $f:\Omega\to \=R^{\geq 0}$,
\begin{equation}\label{eq:approx-tensor-ent}
    \Ent[\mu]{f}\leq C\sum\limits_{v\in V} \E[\mu]{\Ent[\mu]{f\mid \sigma_{V\setminus \{v\}}}}.
\end{equation} 
Approximate tensorization of entropy can also be interpreted as a notion of closeness to a product distribution. In particular,  product distribution satisfies $1$-approximate tensorization of entropy.

For approximate subadditivity of entropy, such notion was known to be equivalent to a Brascamp-Lieb type inequality~\cite{carlen2004sharp,carlen2009subadditivity,barthe2011correlation}, and was also studied under the name of \emph{entropic independence} in a literature of high-dimensional expanders~\cite{anari2022entropic}.  Formally, a distribution $\mu$ over $\Omega\subseteq [q]^V$ is said to satisfy approximate subadditivity of entropy with constant $C$ if for all $f:\Omega\to \=R^{\geq 0}$,
\begin{equation}\label{eq:approx-subadditivity-ent}
    \Ent[\mu]{f}\leq C\sum\limits_{v\in V} \E[\mu]{\Ent[\mu]{f\mid \sigma_v}}.
\end{equation} 

Both notions of approximate tensorization of entropy and approximate subadditivity of entropy can be generalized under the notion of \emph{general block factorization} of entropy. See~\cite{blanca2022mixing} for a detailed discussion of these notions. Note that in contrast to the vertex-wise criteria established in \eqref{eq:approx-tensor-ent} and \eqref{eq:approx-subadditivity-ent}, our criteria in \Cref{def:color-entropy-tensorization} and \Cref{def:color-entropy-subadditivity} are stated with respect to global color classes. 
\end{remark}

\subsubsection{Optimal Mixing through Inductive Entropy Inequalities}
For the second part of our contribution, we verify the ACTE and ACSE criteria for specific graph classes. As concrete applications, we establish an $O(q^2(\log n+\log \log q))$ mixing time upper bound for the WSK dynamics on chordal and outerplanar graphs with $n$ vertices. The required number of colors is optimal down to the colorability threshold, effectively bypassing the irreducibility barrier of Glauber dynamics.

We also prove a complementary lower bound, demonstrating that our upper bounds are worst-case optimal with respect to the graph size $n$.

We establish the ACTE and ACSE criteria via a novel inductive method. Conceptually, this approach functions as a graph-theoretic analogue to local-to-global arguments for proving high-dimensional functional inequalities, a framework we detail later in the technical overview.

We first present our result on chordal graphs. A graph is called \emph{chordal} if it has no induced cycle of length at least four. For a graph $G$, we use $\omega(G)$ to denote its \emph{clique number}, which is the size of the maximum clique in $G$. We establish the following theorem.

\begin{theorem}[Rapid mixing of the WSK dynamics on chordal graphs]\label{thm:WSK-rapid-mixing-chordal}
Let $G$ be a chordal graph with $n$ vertices and let $q\geq \max\{2,\omega(G)\}$. The WSK dynamics for proper $q$-colorings of $G$ mixes in time $O(q^2(\log n+\log \log q))$.
\end{theorem}

    In the context of sampling $q$-colorings of chordal graphs uniformly at random via Markov chains, Heinrich~\cite{Heinrich2020glauber} investigated a closely related chain during their study of Glauber dynamics. As an intermediate step, Heinrich analyzed a variant of flip dynamics, establishing a coupling-based $O_q(n^2)$ mixing bound. In contrast, we analyze the WSK dynamics, which serves as a global counterpart of the flip dynamics. Also, our analysis pursues an entirely different approach by bounding the modified log-Sobolev constant, which ultimately allows us to establish a mixing time bound optimal in $n$. 

Important subclasses of chordal graphs include trees and line graphs of trees. From \Cref{thm:WSK-rapid-mixing-chordal}, we immediately establish the following two corollaries.

\begin{corollary}[Rapid mixing of the WSK dynamics on trees]\label{cor:vertex-coloring-main}
Let $G$ be a tree with $n$ vertices and let $q\geq 2$. The WSK dynamics for proper $q$-colorings of $G$ mixes in time $O(q^2(\log n+\log \log q))$.
\end{corollary}

\begin{corollary}[Rapid mixing of the WSK dynamics for edge colorings on trees]\label{cor:edge-coloring-main}
Let $G$ be a tree with $n$ vertices and maximum degree $\Delta$. Let $q\geq \max\{2,\Delta\}$. The WSK dynamics for proper $q$-edge-colorings of $G$ mixes in time $O(q^2(\log n+\log \log q))$.
\end{corollary}

Previously, for the problem of uniformly sampling colorings on trees with Markov chains, the most studied chain is the Glauber dynamics. It has been shown that the mixing time of Glauber dynamics on trees with maximum degree $\Delta$ is $n^{\Theta({(1+\Delta/(q\log\Delta))})}$ for any $q\geq 3$ and the mixing time undergoes a phase transition at the reconstruction threshold when $q=\Delta(1+o_\Delta(1))/\ln \Delta$~\cite{goldberg2010mixing,tetali2010phase,lucier2011glauber}. In contrast, our result shows that the WSK dynamics has a mixing time optimal with respect to $n$ whenever $q\geq 3$.

The mixing time of Glauber dynamics for edge colorings on trees has also been studied~\cite{delcourt2020glauber,carlson2025optimal}. In ~\cite{carlson2025optimal}, it has been shown that the Glauber dynamics has the optimal relaxation time $O_{q}(n)$ whenever $q\geq \Delta+2$, but they can only prove near-linear mixing time on $\Delta$-regular trees. In contrast, our result shows that the WSK dynamics has a mixing time optimal with respect to $n$ whenever $q\geq \Delta$, which is even beyond the irreducibility threshold of Glauber dynamics.

\begin{remark}[Optimal mixing of global chains beyond strong spatial mixing]
For a long time, the tractability of the sampling problem has been conjectured to be correlated with a decay of correlation property called strong spatial mixing (SSM)~\cite{weitz06counting,bayati2007simple,sinclair12approximation,li2013correlation, lu2013improved,yin2013approximate,gamarnik2015strong,patel2017deterministic,liu2025correlation}. For colorings on trees, strong spatial mixing has been established under the regime of $q\geq \Delta+3$~\cite{chen2023strong}. For edge colorings on trees, strong spatial mixing has only been proven under the regime $q>3\Delta$~\cite{chen2025decay}. Very recently, Chen et al.~\cite{chen2026edgetilting} have shown a sharp mixing bound for the Swendsen--Wang dynamics for the ferromagnetic Ising model with an external field on bounded-degree graphs, which they identify as the first such bound for a classical global Markov chain beyond the SSM regimes. Our work also fits into this category as we achieve an optimal threshold far beyond the currently known SSM regime for edge colorings. Interestingly, on the technical side, both the paper by Chen et al.~\cite{chen2026edgetilting} and our paper use the interpretation of a well-known chain under a new localization scheme.
\end{remark}

The second application of our technique is an optimal mixing result for outerplanar graphs. An outerplanar graph is a graph that has a planar drawing in which all vertices belong to the outer face of the drawing. It is well known that outerplanar graphs are 3-colorable. We prove the following mixing time bound for the WSK dynamics on this class of graphs.

\begin{theorem}[Rapid mixing of the WSK dynamics on outerplanar graphs]\label{thm:WSK-rapid-mixing-outerplanar}
Let $G$ be an outerplanar graph with $n$ vertices. Let $q\geq 3$. The WSK dynamics for proper $q$-colorings of $G$ mixes in time $O(q^2(\log n+\log \log q))$.
\end{theorem}

Finally, we prove a complementary lower bound for the WSK dynamics on path graphs, showing that the upper bounds on mixing times are asymptotically optimal with respect to the size of the graph: 

\begin{theorem}[Lower bound on the mixing time of the WSK dynamics]\label{thm:WSK-rapid-mixing-lower}
Let $G$ be a path graph with $n$ vertices and let $q\geq 3$. The WSK dynamics for proper $q$-colorings of $G$ has mixing time $\Omega_{q}(\log n)$.
\end{theorem}

\subsection{Technical Overview}
To bound the modified log-Sobolev constant and consequently the mixing time of the WSK dynamics, we interpret the dynamics through the lens of localization schemes~\cite{CE22}. While standard Glauber dynamics aligns with a \emph{coordinate-by-coordinate scheme} revealing one random vertex at a time, we introduce a \emph{colorwise localization scheme} that reveals an entire uniformly chosen color class at once. Under this framework, our new approximate colorwise tensorization of entropy (ACTE, \Cref{def:color-entropy-tensorization}) and approximate colorwise subadditivity of entropy (ACSE, \Cref{def:color-entropy-subadditivity}) criteria correspond to entropy conservation throughout the entire localization process and at a single step, respectively. This structural interpretation naturally yields \Cref{thm:tensorization-of-entropy-implies-mixing} and \Cref{thm:subadditivity-implies-tensorization}.

Establishing the vertex-wise approximate subadditivity of entropy criterion typically reduces to bounding the maximal eigenvalue of an influence matrix using the theory of high-dimensional expanders. However, this spectral approach does not seem directly applicable for ACSE, which operates on a global color class. To overcome this, we leverage the color symmetry inherent to the ACSE condition and develop an induction method to prove these global entropy inequalities. This gives a direct graph-theoretic route to proving global entropy inequalities, complementary to the usual local-to-global mechanisms from high-dimensional expansion.
Specifically:\begin{itemize}\item \textbf{Chordal Graphs} (\Cref{thm:WSK-rapid-mixing-chordal}): By definition, every nonempty chordal graph contains a simplicial vertex. Crucially, removing a simplicial vertex preserves the uniformity of the marginal distribution over proper colorings on the residual graph. This property enables an inductive proof of the entropy inequality via the law of total entropy.\item \textbf{Outerplanar Graphs} (\Cref{thm:WSK-rapid-mixing-outerplanar}): The induction generalizes by replacing the simplicial vertex with a pinned cycle (a cycle with two pinned vertices). Because ACSE for pinned cycles cannot be established via direct induction, we introduce a comparison argument by comparing the distribution on cycles against that on paths from removing a uniform random edge, which may be of independent interest. Furthermore, because pinned cycles yield a larger ACSE constant, directly invoking \Cref{thm:subadditivity-implies-tensorization} would introduce an exponential dependency with respect to $q$ in the mixing time. We circumvent this by carefully deriving the ACTE constant directly, an approach conceptually analogous to an average-case local-to-global theorem~\cite{anari2022optimal}.\end{itemize}

Finally, to prove the mixing time lower bound (\Cref{thm:WSK-rapid-mixing-lower}) for the WSK dynamics on path graphs, we employ a distinguishing statistic that tracks the empirical frequency of edges exhibiting a specific pair of adjacent colors. By proving that the variance of this local pattern decays too slowly under the global WSK transition matrix, we leverage standard bounds on the total variation distance to establish an $\Omega_{q}(\log n)$ lower bound.

\subsection{Organization}

The paper is organized as follows:
\begin{itemize}
    \item In \Cref{sec:prelim}, we introduce several necessary preliminaries and notation.
    \item In \Cref{sec:localization}, we interpret the WSK dynamics under a new colorwise localization scheme, and justify the first part of our contribution, namely proving \Cref{thm:tensorization-of-entropy-implies-mixing,thm:subadditivity-implies-tensorization}. 
    \item In \Cref{sec:induction}, we introduce inductive methods for proving the desired entropy inequalities for specific types of graphs, proving \Cref{thm:WSK-rapid-mixing-chordal,thm:WSK-rapid-mixing-outerplanar}.
    \item In \Cref{sec:lower-bound}, we give a lower bound on the mixing time of the WSK dynamics for path graphs, proving \Cref{thm:WSK-rapid-mixing-lower}.
    \item In \Cref{sec:conclusions}, we summarize our contributions and outline potential future directions.
\end{itemize}

\newpage

\section{Preliminaries and Notation}\label{sec:prelim}

\subsection{Graph Theory Basics}

Throughout the paper, all graphs are finite and simple.  For a graph $G=(V,E)$ and a vertex $v\in V$, we write
$N_G(v) \defeq \{u\in V:\{u,v\}\in E\}$ to denote the neighborhood of $v$.   For $U\subseteq V$, let $G[U]$
denote the subgraph of $G$ induced by $U$.  We also write $G-v$ as a
shorthand for $G[V\setminus\{v\}]$.

A set $K\subseteq V$ is called a \emph{clique} if every two distinct vertices
in $K$ are adjacent.  The \emph{clique number} of $G$, denoted by $\omega(G)$,
is the maximum size of a clique in $G$:
\[
    \omega(G) \defeq \max\{|K|: K\subseteq V \text{ is a clique of } G\}.
\]

We then introduce two special kinds of graphs--chordal graphs and outerplanar graphs.

\subsubsection{Chordal Graphs}
A \emph{chord} of a cycle $C$ in $G$ is an edge joining two non-consecutive
vertices of $C$.  A graph $G$ is called \emph{chordal} if every cycle of length
at least four has a chord.  Equivalently, $G$ is chordal if it has no
induced cycle of length at least four.

A vertex $v\in V$ is called \emph{simplicial} if its neighborhood forms a
clique, that is, if $G[N_G(v)]$ is a complete graph.  In particular,
isolated vertices and vertices of degree one are simplicial.

We will only use the following standard facts about chordal graphs.

\begin{fact}[Standard facts about chordal graphs]
\label{fact:chordal-facts}
Let $G=(V,E)$ be a chordal graph. The following holds:
\begin{enumerate}
    \item Every induced subgraph of $G$ is chordal;\label{item:chordal-1}
    \item If $V\neq\emptyset$, then $G$ has a simplicial vertex;\label{item:chordal-2}
    \item If $G'$ is a proper complete induced subgraph of $G$, then $G$ has a simplicial vertex outside $G'$.\label{item:chordal-3}
\end{enumerate}
\end{fact}

The \emph{line graph} of $G$, denoted by $L(G)$, is the graph whose vertex set
is $E$ and in which two distinct vertices $e,e'\in E$ are adjacent
exactly when the corresponding edges of $G$ share an endpoint.
We will use the following elementary observation for trees.

\begin{fact}
\label{fact:line-graph-tree}
If $T$ is a tree with maximum degree $\Delta$, then $L(T)$ is chordal
and
\[
    \omega(L(T))=\Delta.
\]
\end{fact}

\subsubsection{Outerplanar Graphs}

A graph $G$ is \emph{outerplanar} if the graph has a planar embedding $\Gamma$ such that one face (``an outer face'') is
incident to all vertices. The weak dual of an
outerplanar graph $G$, denoted $G^*$
is the dual graph with the vertex corresponding to the outer face removed.

For the purpose of this paper, we only need the following fact for outerplanar graphs.

\begin{fact}\label{fact:outerplanar}
For any 2-connected outerplanar graph $G$, its weak dual $G^*$ is a tree.
\end{fact}

We transform \Cref{fact:outerplanar} into the form of the following corollary, which will be useful for our purpose.

\begin{corollary}[Structural Decomposition of Outerplanar Graphs]\label{cor:structural-decomposition}
    Every finite outerplanar graph $G$ admits a sequence of subgraphs
    \[
    \emptyset = G_0 \subset G_1 \subset \dots \subset G_m = G
    \]
    such that $G_i = G_{i-1} \cup J_i$ for every $i \geq 1$, 
    where the intersection graph $H_i = G_{i-1} \cap J_i$ satisfies one of the following conditions:
    \begin{enumerate}
        \item $J_i$ is an isolated vertex, and $H_i = \emptyset$;\label{item:outplanar-1}
        \item $J_i$ is a single edge, and $H_i$ consists of a single vertex;\label{item:outplanar-2}
        \item $J_i$ is a cycle, and $H_i$ consists of a single edge.\label{item:outplanar-3}
    \end{enumerate}
\end{corollary}

\begin{proof}
    We construct the sequence by fixing an outerplanar embedding of $G$ and processing each connected component independently. For each component, we traverse its block-cut tree starting from an arbitrary root vertex, adding this initial root via \Cref{item:outplanar-1}.

    We then process the remaining blocks of the component in a top-down (parent-before-child) topological order. Let $z \in V(G_{i-1})$ be the cut vertex connecting a newly visited block $B$ to the already-constructed graph.
    
    If $B$ is a bridge, we add it via \Cref{item:outplanar-2}, as it shares exactly the single vertex $z$ with $G_{i-1}$.

    If $B$ is a 2-connected block, its weak dual $B^*$ is a tree by \Cref{fact:outerplanar}. Choose an edge $e_0 \in E(B)$ that is incident to $z$ and lies on the outer face. We first add $e_0$ via \Cref{item:outplanar-2}. Next, we root the weak dual tree $B^*$ at the unique internal face of $B$ containing $e_0$. We then add the internal facial cycles of $B$ following a top-down tree traversal of $B^*$. 
    
    When a facial cycle is added, it acts as a child in $B^*$, meaning it shares exactly one edge with its parent face in the weak dual. Because the embedding is outerplanar, this cycle cannot share any other vertices or edges with previously constructed subgraphs (otherwise, the embedding would be violated by a crossing chord or an enclosed vertex). Thus, each facial cycle is added via \Cref{item:outplanar-3}. This constructs the required sequence.
\end{proof}

\subsection{Colorings and Edge Colorings}
\label{subsec:coloring-distributions}

For an integer $q\geq 1$, write $[q]\defeq\{1,2,\ldots,q\}$.  Let
$G=(V,E)$ be a finite simple graph. We then define colorings and edge colorings.

\subsubsection{Vertex Colorings} A $q$-(vertex) coloring of $G$ is a map $\sigma:V\to [q]$. We say $\sigma$ is a \emph{proper coloring} if adjacent vertices receive distinct colors, that is, $\sigma_u\neq \sigma_v$ for every $\{u,v\}\in E$. We denote the set of proper $q$-colorings of $G$ by $\Omega_{G,q}$. Whenever $\Omega_{G,q}\neq\emptyset$, we use $ \mu_{G,q}$ to denote the uniform distribution over $\Omega_{G,q}$.

For a $q$-coloring $\sigma\in [q]^V$ and $W\subseteq V$, we write $\sigma_W$ for the restriction of $\sigma$ to $W$. When $W=\{v\}$ for some $v\in V$, we write $\sigma_v$ instead of $\sigma_{\{v\}}$ for the ease of notation.

Finally, for a $q$-coloring $\sigma\in [q]^V$ and a color $c\in [q]$, define the color class of color $c$ by
\[
    I_c(\sigma)
    \defeq
    \{v\in V:\sigma_v=c\}.
\]
Also, for a $q$-coloring $\sigma\in [q]^V$ and a nonempty set of colors $S\subseteq [q]$, define the color class of the color set $S$ by
\[
I_S(\sigma)\defeq \tp{I_c(\sigma)}_{c\in S}.
\]

\subsubsection{Edge Colorings}
A $q$-edge coloring of $G$ is a map $\sigma:E\to [q]$. We say $\sigma$ is a \emph{proper edge coloring} if adjacent edges receive distinct colors, that is, $\sigma_{e}\neq \sigma_{e'}$ whenever $e\neq e'$ and $e\cap e'\neq \emptyset$. We denote the set of proper $q$-edge colorings of $G$ by $\Omega^{\edge}_{G,q}$. Whenever $\Omega^{\edge}_{G,q}\neq\emptyset$, we use $ \mu^{\edge}_{G,q}$ to denote the uniform distribution over $\Omega^{\edge}_{G,q}$. We remark that edge colorings on $G$ can be viewed as vertex colorings on $L(G)$, the line graph of $G$.




\subsection{Markov Chain Mixing times, Entropy, and Log-Sobolev Inequalities}
\label{subsec:MC-basics}

Throughout the paper, all logarithms are natural unless otherwise specified.

\subsubsection{Markov Chains and Mixing Times}
Let $\Omega$ be a finite state space, and $(X_t)_{t \ge 0}$ be a Markov chain over $\Omega$ with transition matrix $P$. 
\begin{itemize}
    \item The Markov chain is \emph{irreducible} if for all $x, y \in \Omega$, there exists $t \ge 0$ with $P^t(x,y) > 0$;
    \item The Markov chain is \emph{aperiodic} if $\gcd\{t \ge 1:P^t(x,x)>0\} = 1$ for all $x \in \Omega$.
\end{itemize}
The fundamental theorem of Markov chains shows that the Markov chain $(X_t)_{t \ge 0}$ has a unique \emph{stationary distribution}, i.e. a distribution $\mu$ satisfying $\mu P = \mu$, if $(X_t)_{t \ge 0}$ is irreducible and aperiodic.

Let $(X_t)_{t \ge 0}$ be a Markov chain with transition matrix $P$ and stationary distribution $\mu$. The \emph{mixing time} measures the speed of the Markov chain $(X_t)_{t \ge 0}$ converging to its stationary distribution $\mu$, which is formally defined as
\begin{align*}
T_{\mathrm{mix}}\defeq\max_{x \in \Omega} \min\set{t: d_{\mathrm{TV}}\tp{P^t(x,\cdot),\mu} < \frac{1}{4} }.
\end{align*}

A Markov chain $P$ is \emph{time reversible} with respect to the distribution $\mu$ if it satisfies the {detailed balance equation}:
\begin{align*}
    \forall x,y \in \Omega, \quad \mu(x)P(x,y) = \mu(y)P(y,x).
\end{align*}
This implies that $\mu$ is a stationary distribution of $P$. 
We assume throughout that all Markov chains considered in this paper are time reversible.

\subsubsection{Functional Inequalities}
Let $\mu$ be a distribution over a finite state space $\Omega$, and $(X_t)_{t \ge 0}$ be a Markov chain over $\Omega$ with the transition matrix $P$ and the stationary distribution $\mu$. The \emph{expectation}, \emph{variance}, \emph{entropy}, and the inner product are defined as follows:
\begin{itemize}
\item Expectation: For all $f:\Omega \to \mathbb{R}$, $\E[\mu]{f} \defeq \sum_{x \in \Omega} \mu(x) f(x)$;
\item Variance: For all $f:\Omega \to \mathbb{R}$, $\Var[\mu]{f} \defeq \E[\mu]{f^2} - \tp{\E[\mu]{f}}^2$;
\item Entropy: For all $f:\Omega \to \mathbb{R}_{\geq 0}$, $\Ent[\mu]{f} \defeq \E[\mu]{f\log f}-\E[\mu]{f} \log \E[\mu]{f}$;
\item Inner product: For all $f,g:\Omega \to \mathbb{R}$, $\inner{f}{g}_\mu \defeq \E[\mu]{f \cdot g}$.
\end{itemize}
In particular, we will use the convention $0 \log 0 = 0$ in the definition of entropy. We will sometimes omit the subscript $\mu$ when the distribution is clear from context. In our paper, we will mainly focus on the entropy functional and will consistently deal with conditional entropy. The conditional entropy is formally defined as
\[
\Ent[\mu]{f\mid \+F}=\E[\mu]{f\log f\mid \+F}-\E[\mu]{f\mid \+F}\log \E[\mu]{f\mid \+F},
\]
viewed as an $\+F$-measurable random variable. The following \emph{law of total entropy} will be heavily used throughout the paper:
\[
\Ent[\mu]{f}=\Ent[\mu]{\E[\mu]{f\mid \+F}}+\E[\mu]{\Ent[\mu]{f\mid \+F}}.
\]

Functional inequalities are fundamental tools for quantifying the rate at which a Markov chain converges to its stationary distribution. Specifically, the Poincar\'e inequality captures the decay of variance to stationarity, while the modified log-Sobolev inequality captures the typically faster decay rate of the relative entropy (Kullback-Leibler divergence).

\begin{definition}[Poincar\'e inequality and modified log-Sobolev inequality]
    Let $(X_t)_{t \ge 0}$ be a Markov chain over $\Omega$ with the transition matrix $P$ and stationary distribution $\mu$. \begin{enumerate}
        \item We say the \emph{Poincar\'e inequality} holds with constant $\lambda > 0$ if
        \begin{align*}
        \forall f:\Omega \to \mathbb{R},\quad \lambda \Var[\mu]{f} \le \+E_P(f,f),
        \end{align*}
        where $\+E_P(f,g) \defeq \inner{f}{(I-P)g}_\mu$ is the \emph{Dirichlet form}. 
        \item We say the \emph{modified log-Sobolev inequality} holds with constant $\rho > 0$ if
        \begin{align*}
            \forall f:\Omega \to \mathbb{R}_{\geq 0},\quad \rho \Ent[\mu]{f} \le \+E_P(f,\log f).
        \end{align*}
    \end{enumerate}
\end{definition}

By the Courant-Fischer theorem, the Poincar\'e inequality holds with constant $\lambda$ if and only if $\lambda \le 1-\lambda_2(P)$, where $\lambda_2(P)$ is the second largest eigenvalue of $P$, and $1-\lambda_2(P)$ is the \emph{spectral gap}. 

Suppose the eigenvalues of $P$ are non-negative. The \emph{relaxation time} $T_{\mathrm{rel}}$ is defined as the inverse of the spectral gap, specifically $T_{\mathrm{rel}} = \frac{1}{1-\lambda_2(P)}$. By~\cite[Theorem 12.4]{levin2017markov}, the mixing time of $P$ can be bounded as follows:
 \begin{align*}
    T_{\mathrm{mix}} \le T_{\mathrm{rel}} \cdot \log \frac{4}{\mu_{\min}}, \quad \text{where $\mu_{\min}\defeq\min_{x \in \Omega} \mu(x)$.}
 \end{align*}
With \emph{modified log-Sobolev inequality}, a better mixing time bound can be achieved~\cite{bobkov2006modified}:
\begin{equation}\label{eq:ent-decay-implies-mixing}
T_{\mathrm{mix}} \le \frac{1}{\rho} \tp{\log \log \frac{1}{\mu_{\min}} + 3}.
\end{equation}

\subsection{Localization Schemes}
\label{subsec:information-theoretic}
In this subsection, we introduce the notion of localization schemes, which is a framework for proving rapid mixing of Markov chains, introduced in~\cite{CE22}. In ~\cite{CE22}, localization schemes are defined as a mapping from a target measure to a probability-measure-valued stochastic process. Here we instead adopt the information-theoretic definition of localization schemes~\cite{ElAlaoui2022information}, which was also used in previous works~\cite{chen2025rapid,chen2025rapidrandom,chen2026edgetilting} for sampling.

\begin{definition}[Localization processes and localization schemes~\cite{CE22,chen2025rapid}]\label{def:localization-schemes}
Let $\mu$ be a distribution on a finite space $\Omega$. A \emph{localization process} for $\mu$ is formalized by an information space and a pair of continuous-time processes, defined as follows:
\begin{itemize}
    \item The \emph{information space} $(\mathsf{E}_t)_{t\in[0,1]}$ is a family of standard Borel spaces, where each $\mathsf{E}_t$ represents the space of possible observations available at time (or signal level) $t$, together with a measurable output map
    \[
    \phi:\mathsf E_1\to\Omega.
\]
Each element $\iota\in \mathsf{E}_t$ for $t\in [0,1]$ is associated with a probability distribution $\nu^{(t)}_{\iota}$ over $\Omega$. Formally, this space satisfies:
    \begin{itemize}
        \item (Initialization) $\mathsf{E}_0 = \{\perp\}$, where $\perp$ denotes the null-information state, and $\nu^{(0)}_{\perp}=\mu$. 
        \item (Localization) For each element in the information space $E_1$ at time $1$, the associated probability distribution collapses into a Dirac measure specified by the output map, i.e.,
\[
    \nu^{(1)}_\iota=\delta_{\phi(\iota)},
    \qquad \forall \iota\in\mathsf E_1.
\]
In the special case \(\mathsf E_1=\Omega\), it suffices to take the output map $\phi$ to be the identity map.
    \end{itemize}

    \item The \emph{denoising process} $(Y_t)_{t\in[0,1]}$ is a Markovian stochastic process taking values in \(\mathsf E_t\). The localized measure at time \(t\) is defined as
\[
    \nu_t\defeq \nu^{(t)}_{Y_t}.
\]
We require the following martingale property: for every \(A\subseteq\Omega\) and every \(0\le s\le t\le 1\),
\[
    \E{\nu^{(t)}_{Y_t}(A)\mid Y_s}
    =
    \nu^{(s)}_{Y_s}(A).
\]
Equivalently, \((\nu_t(A))_{t\in[0,1]}\) is a martingale with respect to the information carried by the observations \((Y_t)_{t\in[0,1]}\). We remark that for every $t\in [0,1]$ and every information state $\iota$ in the support of $Y_t$, the following holds as a consequence of the martingale property:
\[
\nu^{(t)}_{\iota}=\mathrm{Law}(\phi(Y_1)\mid Y_t=\iota).
\]
    
    \item The \emph{noising process} $(X_t)_{t\in[0,1]}$
    is the time reversal of the denoising process, meaning
    \[
        (X_t)_{t\in[0,1]}
        \overset{d}=
        (Y_{1-t})_{t\in[0,1]}.
    \]
\end{itemize}

Because these processes strictly determine one another via time-reversal, we conventionally identify the localization process entirely by its denoising component, $(Y_t)_{t\in [0,1]}$. This definition naturally extends to discrete time by replacing $[0,1]$ with $\{0,1,\ldots,T\}$.

Finally, a \emph{localization scheme} $\+L$ on $\Omega$ is a map assigning a localization process to each distribution $\mu$ on $\Omega$. 
\end{definition}

\begin{remark}[The information space]\label{rem:information-space}
A key distinction of our definition is the explicit formalization of the information space $(\mathsf{E}_t)_{t\in[0,1]}$. Prior sampling literature leveraging this information-theoretic framework~\cite{chen2025rapid,chen2025rapidrandom,chen2026edgetilting} primarily focuses on \emph{field denoising processes} for set systems. In those settings, the intermediate observation $Y_t$ is a partially revealed set that can be defined directly on $\Omega$, hence the definition of the information space (as well as the output map $\phi$) was not necessary.

In contrast, for the coordinate-by-coordinate localization process in~\cite{CE22}, the intermediate observation $Y_t$ takes values in a space that encodes both the subset of revealed coordinates and their specific configuration. Furthermore, in the context of the proximal sampler, the clean state might reside in $Y_1 \in \Omega \subseteq \{\pm 1\}^n$, whereas the intermediate observation is a Gaussian vector $Y_t\in \mathbb{R}^r$ for $0 < t < 1$. Thus, formally incorporating the information space into the definition of the scheme is mathematically essential to capture these broader classes of processes.
\end{remark}

Given a target distribution $\mu$ defined on a finite state space $\Omega$, a denoising process naturally induces a Markov chain with stationary distribution $\mu$.


\begin{definition}[Markov chain associated with a denoising process]\label{def:associated-dynamics}
Let $\mu$ be a distribution defined on a state space $\Omega$.  Let $(Y_t)_{t\in [0,1]}$ be a denoising process with respect to $\mu$. Fix any $\theta\in [0,1]$, the \emph{down-up} chain $P_{1\leftrightarrow \theta}$ induced by $(Y_t)_{t\in [0,1]}$ is a Markov chain on $\Omega$ defined as follows.
Given the current state $x_t\in \Omega$, the next state $x_{t+1}$ is generated by 
\begin{itemize}
    \item (downward step $D_{1\to \theta}$) given the current state $x_{t}\in\Omega$, an intermediate state $x'$ is generated by (note that $x'\in \mathsf{E}_{\theta}$)
    \[
    x'\sim \mathrm{Law}(Y_\theta\mid \phi(Y_1)=x_t);
    \]
    \item (upward step $U_{\theta\to 1}$) then sample the next state
    \[
    x_{t+1}\sim \mathrm{Law}(\phi(Y_1)\mid Y_\theta=x').
    \]
\end{itemize}
This chain is reversible with stationary distribution $\mu$.
\end{definition}


\subsubsection{Conservation of Entropy}

The reason one considers this formulation of localization scheme is that it transforms the analysis for the mixing time of Markov chains to the analysis of the noising-denoising stochastic process itself. We need to introduce the notion of approximate conservation of entropy.

\begin{definition}[Approximate Conservation of Entropy]\label{def:approx-conservation-var}
Given a denoising process $(Y_t)_{t \in [0,1]}$ and a parameter $\theta \in [0,1]$, we say that $(Y_t)_{t \in [0,1]}$ satisfies \emph{$\kappa$-approximate conservation of entropy up to time $\theta$} if for all functions $f:\Omega \to \mathbb{R}_{\geq 0}$,
\begin{align*}
\Ent[]{f(\phi(Y_1))} \le \kappa\cdot\E{\Ent[]{f(\phi(Y_1))\mid Y_{\theta}}}.
\end{align*}
\end{definition}

The following lemma, combined with \eqref{eq:ent-decay-implies-mixing}, shows that conservation of entropy throughout the localization scheme implies rapid mixing of the associated dynamics.

\begin{lemma}[{\cite[Lemma 3.11]{chen2025rapid}}]\label{lem:AC-var-to-var-decay}
   If the denoising process $(Y_t)_{t \in [0,1]}$ satisfies $\kappa$-approximate conservation of entropy up to time $\theta$, 
then the modified log-Sobolev constant of the induced down-up chain $P_{1 \leftrightarrow \theta}$ is at least $1/\kappa$.
\end{lemma}

\newpage

\section{Colorwise Localization Scheme for the WSK dynamics}\label{sec:localization}

In this section, we will present a new (discrete) localization scheme for multi-spin systems, which we call the colorwise localization scheme, whose associated dynamics coincides with the WSK dynamics. We will show that under this interpretation of colorwise localization scheme, the ACTE (\Cref{def:color-entropy-tensorization}) and ACSE (\Cref{def:color-entropy-subadditivity}) criteria correspond to the entropy conservation of the localization process. Consequently, we will prove \Cref{thm:tensorization-of-entropy-implies-mixing,thm:subadditivity-implies-tensorization}.

Compared to the classical coordinate-by-coordinate localization scheme (whose associated dynamics is the Glauber dynamics or block dynamics), which reveals the color of a random coordinate at each step, the colorwise localization scheme randomly chooses a remaining color at each step and reveals all coordinates of that color simultaneously.

\begin{definition}[Colorwise localization scheme and associated Markov chain]\label{def:localization-color}
Let $q \geq 2$ and let $\mu$ be a distribution on a finite state space $\Omega \subseteq [q]^V$. For any configuration $\sigma \in \Omega$, let $S(\sigma) \defeq \{(v, \sigma_v) \mid v \in V\}$ denote its representation as a set of vertex-color pairs. 

The \emph{colorwise localization scheme} $\+L^{\mathrm{col}} = \+L^{\mathrm{col}}(\mu)$ is the discrete-time localization process $(Y_t)_{t=0}^q$ generated by the following denoising procedure:
\begin{enumerate}
    \item Sample a target configuration $\sigma \sim \mu$.
    \item Sample a permutation $p = (p_1, p_2, \dots, p_q)$ of the color set $[q]$ uniformly at random \emph{on the fly} (i.e., $p_t$ is revealed at time $t$).
    \item For each intermediate step $t \in \{0, 1, \dots, q\}$, the intermediate state $Y_t$ reveals the first $t$ colors of the permutation alongside the subset of vertices assigned those colors. Formally, $Y_t \in \mathsf{E}_t$ is defined as the tuple:
    \[
        Y_t = \left( (p_1, \dots, p_t), \{(v,c) \in S(\sigma) \mid p^{-1}(c) \leq t \} \right).
    \]
    The output map $\phi:\mathsf{E}_{q}\to \Omega$ is defined as $\phi(Y_q)=\sigma$.
\end{enumerate}

Note that in the above definition of $Y_t$, the information space $\mathsf{E}_t$ explicitly encodes both the partial configuration and the specific sequence of colors revealed so far. It is direct to verify that the procedure described above satisfies the three required properties (initialization, localization and martingale property) in \Cref{def:localization-schemes} and therefore is indeed a localization scheme.

For a fixed $\theta \in \{0, 1, \dots, q\}$, the associated Markov chain $P_{q\leftrightarrow \theta}$ is defined according to (the natural discrete-time variation of) \Cref{def:associated-dynamics}. 
\end{definition}

An important observation is as follows.

\begin{lemma}
    Let $\mu$ be the uniform distribution over proper $q$-colorings of a $q$-colorable graph $G$. The associated Markov chain $P_{q\leftrightarrow q-2}$ in \Cref{def:localization-color} coincides exactly with the \emph{WSK dynamics} in \Cref{def:WSK-dynamics}.
\end{lemma}

\begin{proof}
    Note that in \Cref{def:localization-color} the permutation is chosen uniformly at random on-the-fly. Therefore at time $\theta=q-2$, only the first $q-2$ colors are revealed. Hence, in the downward step $D_{q\to \theta}$ from a current coloring $\sigma$, revealing $q-2$ colors is equivalent to choosing an unordered pair $\{a,b\}$ of unrevealed colors uniformly. The unrevealed vertex set is exactly
    \[
    W=W_{a,b}(\sigma)\defeq\{v\in V:\sigma_v\in\{a,b\}\}.
\]
In the upward step $U_{\theta\to q}$, one samples uniformly from proper colorings agreeing with $\sigma$ outside $W$ and using only $a,b$ on $W$. Since each component of $G[W]$ is properly two-colored, it has exactly two possible colorings, obtained by either retaining or swapping $a,b$. Different components are independent under the uniform conditional distribution. Thus one step of the associated dynamics $P_{q\leftrightarrow \theta}=D_{q\to \theta}U_{\theta\to q}$ is exactly the WSK update from \Cref{def:WSK-dynamics}.
\end{proof}

Through this interpretation of the WSK dynamics under the colorwise localization scheme, we are now ready to prove the two main theorems in the introduction. This is by noticing that the ACTE criterion corresponds to the entropy conservation of the localization process from time $0$ to $q-2$, while the ACSE criterion corresponds to the one-step entropy conservation of the localization process from time $t$ to $t+1$ for some $0\leq t<q-2$.

\begin{proof}[Proofs of \Cref{thm:tensorization-of-entropy-implies-mixing,thm:subadditivity-implies-tensorization}]
   We first prove \Cref{thm:tensorization-of-entropy-implies-mixing}. Consider the colorwise localization scheme $\+L^{\mathrm{col}} = \+L^{\mathrm{col}}(\mu)=(Y_t)_{t=0}^q$ in \Cref{def:localization-color}. Note that $\sigma\defeq \phi(Y_q)$ is distributed as $\mu$, and that $Y_{q-2}$ is distributed by uniformly revealing $q-2$ distinct color classes from $Y_q$. Therefore, for any $f:\Omega\to \=R^{\geq 0}$, it holds that
     \begin{equation}\label{eq:localization-transform}
     \Ent{f(\phi(Y_q)}=\Ent[\mu]{f},\quad \E{\Ent{f(\phi(Y_q))}\mid Y_{q-2}}=\frac{1}{\binom{q}{2}}\sum\limits_{\{a,b\}\in \binom{[q]}{2}}\E[\mu]{\Ent[\mu]{f\mid I_{[q]\setminus\{a,b\} }}}.
     \end{equation}
     Therefore, by combining  \Cref{def:color-entropy-tensorization} and \Cref{lem:AC-var-to-var-decay}, we have that the modified log-Sobolev constant of the WSK dynamics is at least $2/Cq(q-1)$.
     
     Finally, note that we have $\mu_{\min}\geq q^{-n}$ for uniform proper $q$-colorings, \Cref{thm:tensorization-of-entropy-implies-mixing} holds by \eqref{eq:ent-decay-implies-mixing}.




   We then prove \Cref{thm:subadditivity-implies-tensorization}. We claim that under the assumption of \Cref{thm:subadditivity-implies-tensorization}, for every $0\leq i\leq q-3$ and every test function $f:\Omega\to \=R^{\geq 0}$ it holds that
     \begin{equation}\label{eq:approximate-conservation-of-variance}
      \E{\Ent[][]{f(\phi(Y_q))\mid Y_{i+1}}} \geq  \frac{1}{(q-i)C(q-i)}\E{\Ent{f(\phi(Y_q))\mid Y_i}}.
     \end{equation}
     Then we have
     \[
   \E{\Ent[][]{f(\phi(Y_q))\mid Y_{q-2}}}\geq \prod\limits_{i=3}^{q} \frac{1}{iC(i)}\E{\Ent{f(\phi(Y_q))}\mid Y_0}=\prod\limits_{i=3}^{q} \frac{1}{iC(i)}\Ent{f(\phi(Y_q))}.
    \]
     Then \Cref{thm:subadditivity-implies-tensorization} holds by \eqref{eq:localization-transform}.
     
     It then suffices to show \eqref{eq:approximate-conservation-of-variance}.  We notice that at each time $0\leq i\leq q-2$, the induced distribution $\phi(Y_q)\mid Y_i$ over all remaining unspecified vertices $V'$ is the uniform $(q-i)$-coloring on the subgraph of $G$ induced by $V'$. Therefore for each $0\leq i\leq q-3$,
     we have
     \begin{align*}
        &  \E{\Ent[][]{f(\phi(Y_q))\mid Y_{i+1}}\mid Y_i}\\
        = & \frac{1}{q-i}\sum\limits_{c\in [q]\setminus \{p_1,\dots,p_i\}}\E{\Ent[][]{f(\phi(Y_q))\mid Y_{i+1}}\mid Y_i, p_{i+1}=c}\\
    (\text{ACSE of $\phi(Y_q)\mid Y_i$})\quad    \geq &\frac{1}{(q-i)C(q-i)} \Ent[][]{f(\phi(Y_q))\mid Y_{i}}.
     \end{align*}
     Averaging over all possible $Y_i$ gives \eqref{eq:approximate-conservation-of-variance}, thereby concluding the proof of \Cref{thm:subadditivity-implies-tensorization}.
\end{proof}

\section{Inductive Proofs Establishing Approximate Colorwise Tensorization of Entropy}\label{sec:induction}

In this section, we will show approximate colorwise tensorization of entropy on two specific types of graphs: chordal graphs and outerplanar graphs, therefore proving \Cref{thm:WSK-rapid-mixing-chordal,thm:WSK-rapid-mixing-outerplanar}. Our proofs will apply an inductive approach, building upon the properties of the graphs and the requirement for ACTE (\Cref{def:color-entropy-tensorization}). We start with chordal graphs.

\subsection{Approximate Colorwise Tensorization of Entropy for Chordal Graphs}

We will show that chordal graphs satisfy approximate colorwise tensorization of entropy with constant $1$. 

\begin{lemma}[ACTE on chordal graphs]\label{lem:ACTE-chordal}
 Let $G=(V,E)$ be a chordal graph and let $q\geq \max(3,\omega(G))$. Then $\mu_{G,q}$ satisfies approximate colorwise tensorization of entropy with constant $1$. 
\end{lemma}

Note that by \Cref{thm:tensorization-of-entropy-implies-mixing}, \Cref{lem:ACTE-chordal} implies the proof of \Cref{thm:WSK-rapid-mixing-chordal} (the $q=2$ case mixes trivially in one step), hence the proof of \Cref{cor:vertex-coloring-main} and \Cref{cor:edge-coloring-main}.

To show ACTE on chordal graphs, we will establish the following result of approximate colorwise subadditivity of entropy on chordal graphs.

\begin{lemma}[ACSE on chordal graphs]\label{lem:ACSE-chordal}
    Let $G=(V,E)$ be a chordal graph and let $q\geq \max(3,\omega(G))$. Then $\mu_{G,q}$ satisfies approximate colorwise subadditivity of entropy with constant $1/(q-2)$.
\end{lemma}

Note that combining \Cref{thm:subadditivity-implies-tensorization} with \Cref{lem:ACSE-chordal} directly proves \Cref{lem:ACTE-chordal}.

To prove \Cref{lem:ACSE-chordal}, we leverage the property of chordal graphs to set up a simple induction bounding the entropy decay. We first prove ACSE on a singleton graph, which serves as the base case for our induction.

\begin{lemma}[ACSE on a singleton graph]\label{lem:bounded-entropy-decay-singleton}
    Let $G=(V,E)$ be a singleton graph and let $q\geq 3$. Then $\mu_{G,q}$ satisfies approximate colorwise subadditivity of entropy with constant
    \[
       \frac{\log q}{(q-1)\log {(q-1)}}.
    \]
\end{lemma}

To prove \Cref{lem:bounded-entropy-decay-singleton}, we rely on the following lemma from \cite{gu2023nonlinear}.

\begin{lemma}[{\cite[Proposition 32]{gu2023nonlinear}}]\label{lem:wrong-color-entropy-bound}
Let $q\geq 3$. For any probability vector $\mathbf{p}=(p_1,p_2,\dots,p_q)$ over $[q]$, let $\mathbf{r}=(r_1,r_2,\dots,r_q)$ where $r_i=\frac{1-p_i}{q-1}$. Then
\[
D(\mathbf{r}\Vert \mathrm{Unif}_q)\leq \frac{\log q-\log{(q-1)}}{\log q}D(\mathbf{p}\Vert \mathrm{Unif}_q).
\]
\end{lemma}

\begin{proof}[Proof of \Cref{lem:bounded-entropy-decay-singleton}]
    Let $\mu=\mu_{G,q}$. Because $G$ is a singleton, $\mu=\mathrm{Unif}_q$. Without loss of generality, we assume that $\E{f}=1$. Let $\mathbf{p}=(p_1,p_2,\dots,p_q)$ be a probability vector where $p_i=\frac{f(i)}{q}$; it then holds that $\Ent{f}=D(\mathbf{p}\Vert \mu)$. Also let $\mathbf{r}=(r_1,r_2,\dots,r_q)$ where $r_i=\frac{1-p_i}{q-1}$.

    Note that for each $c\in [q]$ we have
    \[
    \E{f\mid I_c}=\begin{cases}qp_c, & \sigma=c, \\ \frac{q(1-p_c)}{q-1}, & \sigma\neq c, \end{cases}
    \]
    and therefore
    \[
    \Ent{\E{f\mid I_c}}=p_c\log{(qp_c)}+(1-p_c)\log {\frac{q(1-p_c)}{q-1}}.
    \]
    Summing over all $c \in [q]$, we obtain
    \begin{equation}\label{eq:singleton-entropy-relationship}
    \sum\limits_{c\in [q]}\Ent{\E{f\mid I_c}}=D(\mathbf{p}\Vert \mu)+(q-1)D(\mathbf{r}\Vert \mu)\leq\frac{q\log q-(q-1)\log {(q-1)}}{\log q}\Ent{f},
    \end{equation}
    where the last inequality follows from \Cref{lem:wrong-color-entropy-bound} and the fact that $\Ent{f}=D(\mathbf{p}\Vert \mu)$.
    
    By the law of total entropy, we have that for each $c\in [q]$,
    \[
    \Ent{f}=\Ent{\E{f\mid I_c}}+\E{\Ent{f\mid I_c}}.
    \]
    Summing this over all colors yields
    \[
    q\Ent{f}=\sum\limits_{c\in [q]}\Ent{\E{f\mid I_c}}+\sum\limits_{c\in [q]}\E{\Ent{f\mid I_c}}.
    \]
    Combining this with \eqref{eq:singleton-entropy-relationship}, we obtain
    \[
    \Ent{f}\leq \frac{\log q}{(q-1)\log {(q-1)}}\sum\limits_{c\in [q]}\E{\Ent{f\mid I_c}},
    \]
    concluding the proof.
\end{proof}

We are now ready to prove \Cref{lem:ACSE-chordal}, hence concluding the proof of \Cref{thm:WSK-rapid-mixing-chordal}. 

\begin{proof}[Proof of \Cref{lem:ACSE-chordal}]
      Let $\mu=\mu_{G,q}$. We proceed by induction on the size of the graph $|V|$. 
      
      \textbf{Base Case ($|V|=1$):} Note that for each $q\geq 3$,
    \[
    (q-1)\log {(q-1)}-(q-2)\log {q}=\log {q}-(q-1)\log {\frac{q}{q-1}}>\log {q}-1>0.
    \]
    Therefore, the bound $\frac{ \log q}{(q-1)\log(q-1)} < \frac{1}{q-2}$ holds, and the base case follows directly from \Cref{lem:bounded-entropy-decay-singleton}. 
    
    \textbf{Inductive Step ($|V| > 1$):} Because $G$ is chordal, there exists a simplicial vertex $v\in V$. We choose one such $v$ arbitrarily. Note that the induced subgraph $H=G\setminus \{v\}$ is also chordal and $q\geq \omega(H)$. Moreover, since the neighborhood of $v$ is a clique, so its colors are distinct, and the number of available colors for $v$ is $q-|N(v)|$, independent of the coloring. Therefore, every proper coloring of $H$ can be extended to a proper coloring of $G$ in exactly the same number of ways. Therefore, the distribution $\mu_{H,q}$ is exactly the marginal distribution of $\mu$ on $V\setminus \{v\}$. Recall that for each $\sigma\in \Omega_{G,q}$ and $c\in [q]$, we define the indicator configuration $I_c(\sigma)=\{u\in V\mid \sigma_u=c\}$.
    We additionally define the partial indicators: 
    \[
    I_c^{0}(\sigma)\defeq \{u\in V\setminus \{v\}\mid \sigma_u=c\}, \quad B_c(\sigma)\defeq \one{\sigma_v=c}, 
    \]
    so that $I_c$ encodes the exact same information as the pair $(I_c^{0},B_c)$. 

    Applying the law of total entropy, we have for each $c\in [q]$:
    \begin{equation}\label{eq:chordal-induction-law-of-total-entropy}
         \E{\Ent{f\mid I_c}}=\E{\Ent{\E{f\mid \sigma_{V\setminus \{v\}},B_c}\mid I_c}}+\E{\Ent{f\mid \sigma_{V\setminus \{v\}},B_c}},
    \end{equation}
    where in the final term we dropped the conditioning on $I_c$ because the full partial configuration $\sigma_{V\setminus \{v\}}$ together with $B_c$ strictly determines $I_c$.

    We bound the sum over all colors for the two terms on the RHS of \eqref{eq:chordal-induction-law-of-total-entropy} separately. Let $\bar{f}\defeq \E{f\mid \sigma_{V\setminus \{v\}}}$ be the function averaged over $v$. For the first term, note that for each $c\in [q]$:
    \[
    \E{\Ent{\E{f\mid \sigma_{V\setminus \{v\}},B_c}\mid I_c}}=\E{\Ent{\E{f\mid \sigma_{V\setminus \{v\}},B_c}\mid I_{c}^0,B_c}}\geq \E{\Ent{\bar{f}\mid I_c^0}},
    \]
    where the final inequality uses the convexity of entropy and the fact that $B_c$ is independent of $\sigma_{V\setminus \{v\}}$ conditioned on $I^0_c$. Summing over all colors and applying the induction hypothesis to $H$, we have:
    \begin{equation}\label{eq:chordal-induction-first-part}
    \sum\limits_{c\in [q]} \E{\Ent{\E{f\mid \sigma_{V\setminus \{v\}},B_c}\mid I_c}}\geq \E{\sum\limits_{c\in [q]}\Ent{\bar{f}\mid I_c^0}}\geq (q-2)\Ent{\bar{f}}.
    \end{equation}
    
    We now focus on the second term. Because $v$ is simplicial, any valid coloring $\sigma\in \Omega_{G,q}$ must assign distinct colors to the neighborhood $N(v)$.  We claim that:
    \begin{equation}\label{eq:chordal-induction-averaging}
        \sum\limits_{c\in [q]}\E{\Ent{f\mid \sigma_{V\setminus \{v\}}, B_c}\mid \sigma_{V\setminus \{v\}}}\geq (q-2)\Ent{f\mid \sigma_{V\setminus \{v\}}}.
    \end{equation}
    Let $T$ denote the set of colors in $[q]$ \emph{not present} in $\sigma_{N(v)}$. Thus $|T|=q-|N(v)|$, and because $q\geq \omega(G)$, we are guaranteed $|T|\geq 1$. When $|T|= 1$, $f$ is completely determined by $\sigma_{V\setminus \{v\}}$ and \eqref{eq:chordal-induction-averaging} trivially holds. We thus assume $|T|=s\geq 2$. For colors $c\notin T$, we have $B_c=0$ deterministically conditioned on $\sigma_{V\setminus \{v\}}$, and therefore: 
    \begin{equation}\label{eq:chordal-induction-color-present}
    \Ent{f\mid \sigma_{V\setminus \{v\}}, B_c}=\Ent{f\mid \sigma_{V\setminus \{v\}}}, \quad\forall c\notin T.
       \end{equation}
    For colors $c\in T$, conditioning on $\sigma_{V\setminus \{v\}}$ leaves $\sigma_v$ uniformly distributed over $T$. By \Cref{lem:bounded-entropy-decay-singleton} (and noting $(s-1)\log {(s-1)}>(s-2)\log s$ for $s \geq 3$), we have for the case when $s\geq 3$: 
     \begin{equation}\label{eq:chordal-induction-color-not-present}
\sum\limits_{c\in T}\E{\Ent{f\mid \sigma_{V\setminus \{v\}}, B_c}\mid \sigma_{V\setminus \{v\}}}\geq (s-2)\Ent{f\mid \sigma_{V\setminus \{v\}}}.
     \end{equation}
     Note that \eqref{eq:chordal-induction-color-not-present} trivially holds for the case $s=2$.
     
     Combining \eqref{eq:chordal-induction-color-present} and \eqref{eq:chordal-induction-color-not-present} proves \eqref{eq:chordal-induction-averaging}. Averaging over all configurations $\sigma_{V\setminus \{v\}}$, we obtain:
      \begin{equation}\label{eq:chordal-induction-second-part}
     \sum\limits_{c\in [q]}\E{\Ent{f\mid \sigma_{V\setminus \{v\}}, B_c}}\geq (q-2)\E{\Ent{f\mid \sigma_{V\setminus \{v\}}}}.
     \end{equation}

     Finally, combining \eqref{eq:chordal-induction-law-of-total-entropy}, \eqref{eq:chordal-induction-first-part}, and \eqref{eq:chordal-induction-second-part} yields:
     \[
     \sum\limits_{c\in [q]}\E{\Ent{f\mid I_c}}\geq (q-2)\Ent{\bar{f}}+(q-2)\E{\Ent{f\mid \sigma_{V\setminus \{v\}}}}=(q-2)\Ent{f},
     \]
     where the last equality follows directly from the law of total entropy. This concludes the inductive step and the proof.
\end{proof}

\subsection{Approximate Colorwise Tensorization of Entropy for Outerplanar Graphs}

We will show that outerplanar graphs satisfy approximate colorwise tensorization of entropy with a larger constant $1/(1-\log 2)$.

\begin{lemma}[ACTE on outerplanar graphs]\label{lem:ACTE-outerplanar}
 
 Let $G=(V,E)$ be an outerplanar graph and let $q\geq 3$.  Then $\mu=\mu_{G,q}$ satisfies approximate colorwise tensorization of entropy with constant $1/(1-\log 2)$.  
\end{lemma}

Note that by combining with \Cref{thm:tensorization-of-entropy-implies-mixing}, \Cref{lem:ACTE-outerplanar} immediately proves \Cref{thm:WSK-rapid-mixing-outerplanar}.

To establish ACTE on outerplanar graphs, we maintain an inductive approach but modify the reduction strategy applied to chordal graphs. Specifically, instead of eliminating a simplicial vertex, each inductive step consists of removing an ear associated with a leaf face of the planar embedding. To this end, the basic distribution we need to consider is the uniform coloring on a cycle with two vertices pinned. We first establish ACSE for such types of distributions.

\begin{lemma}[ACSE on pinned cycle graphs]\label{lem:ACSE-cycle}
    Let $G=(V,E)$ be a cycle graph with $|V|\geq 3$ and let $q\geq 3$. Fix an edge $e_{\star}=\{u_{\star},v_{\star}\}\in E$ and distinct colors $\xi,\zeta\in [q]$. Let $\mu$ be the uniform measure on proper colorings $\sigma$ of $G$ satisfying $\sigma_{u_{\star}}=\xi$ and $\sigma_{v_{\star}}=\zeta$. Then $\mu$ satisfies approximate colorwise subadditivity of entropy with constant 
    \[
        \frac{1}{(q-2)(1-\log 2)}.
    \]
\end{lemma}

Before trying to establish a proof for \Cref{lem:ACSE-cycle}, we note that directly applying \Cref{thm:subadditivity-implies-tensorization} for the ACSE constant in \Cref{lem:ACSE-cycle} will lead to an ACTE constant that has super-polynomial dependency in $q$. Hence, we will prove the following stronger version of \Cref{lem:ACSE-cycle}, which will turn out to lead to a proof for a polynomial ACTE constant.

\begin{lemma}[refined one-color reveal bound on pinned cycle graphs]\label{lem:refined-ACSE-cycle}
    Let $G=(V,E)$ be a cycle graph with $|V|\geq 3$ and let $q\geq 3$. Fix an edge $e_{\star}=\{u_{\star},v_{\star}\}\in E$ and distinct colors $\xi,\zeta\in [q]$. Let $\mu$ be the uniform measure on proper colorings $\sigma$ of $G$ satisfying $\sigma_{u_{\star}}=\xi$ and $\sigma_{v_{\star}}=\zeta$. Then for any $f:\Omega\to \=R^{\geq 0}$ it holds that
    \begin{equation}\label{eq:refined-entropy-conservation-cycle-constant}
        \Ent[\mu]{f}\leq \frac{1}{(q-2)(1-\log 2)}(H_{+}(f)+(1-\log 2)H_{\emptyset}(f)),
    \end{equation}
    where 
    \[
   H_{\emptyset}(f)\defeq \sum\limits_{c\in [q]}\E[\mu]{\one{I_c=\emptyset}\Ent[\mu]{{f\mid I_c}}},\quad H_{+}(f)\defeq \sum\limits_{c\in [q]}\E[\mu]{\one{I_c\neq \emptyset}\Ent[\mu]{{f\mid I_c}}}.
    \]
\end{lemma}

We remark that \Cref{lem:refined-ACSE-cycle} is indeed stronger than \Cref{lem:ACSE-cycle}. This is because we have
\[
\sum\limits_{c\in [q]}\E[\mu]{\Ent[\mu]{f\mid I_c}}=H_{\emptyset}(f)+H_{+}(f)\geq H_{+}(f)+(1-\log 2)H_{\emptyset}(f).
\]
Therefore it suffices to prove \Cref{lem:refined-ACSE-cycle}.

Note that for pinned cycle graphs the proof does not follow directly from induction. Our proof is based on a comparison with ACSE on chordal graphs when we remove one unpinned edge uniformly at random. We still need a pinned version of chordal ACSE, which follows directly from the same inductive proof of \Cref{lem:ACSE-chordal}.

\begin{lemma}[ACSE on chordal graphs with pinned cliques]\label{lem:ACSE-pinned-chordal}
    Let $G=(V,E)$ be a chordal graph and let $q\geq \max(3,\omega(G))$. Let $K\subseteq V$ be a clique with a fixed proper coloring, and let $\mu$ be the uniform distribution on proper $q$-colorings of $G$ extending that coloring. Then $\mu$ satisfies approximate colorwise subadditivity of entropy with constant $1/(q-2)$.
\end{lemma}
\begin{proof}
    Note that by \Cref{fact:chordal-facts}-(\ref{item:chordal-3}), when not all vertices have color fixed, there exists at least one uncolored simplicial vertex. Also, removing one uncolored simplicial vertex leaves the induced distribution on the remaining graph the same even in the presence of a pinned clique. Therefore, one can apply the same induction as in \Cref{lem:ACSE-chordal} to prove the lemma.
\end{proof}

We are now ready to prove \Cref{lem:ACSE-cycle,lem:refined-ACSE-cycle}.
\begin{proof}[Proof of \Cref{lem:ACSE-cycle,lem:refined-ACSE-cycle}]
Let $n=|V|$. Without loss of generality, we assume that $\E[\mu]{f}=1$. By the law of total entropy, it holds that
\begin{equation}\label{eq:entropy-loss}
H_\emptyset(f)+H_{+}(f)=q\Ent{f}-L(f),
\end{equation}
where 
\[
L(f)\defeq\sum\limits_{c\in [q]}\Ent[\mu]{\E[\mu]{f\mid I_c}}
\]
denotes the sum of entropy loss after revealing one color.

We first try to establish an upper bound on the entropy loss term in \eqref{eq:entropy-loss}, therefore giving a lower bound on $H_\emptyset(f)+H_{+}(f)$, and later see we can refine such method to also upper bound $H_{+}(f)+(1-\log 2)H_{\emptyset}(f)$.  We achieve this by comparing the cycle graph with a path graph obtained by uniformly removing a random edge. We will separate the entropy loss into two parts: the entropy loss bounded by the path graph comparison, and the residual entropy loss due to the cyclic constraint.

Consider the path graph $G_e=(V,E\setminus \{e\})$ obtained from $G$ by removing an arbitrary unpinned edge $e=\{u,v\}\in E'\defeq E\setminus \{e_{\star}\}$. Let $\Omega_e$ be the set of proper colorings $\sigma$ of $G_e$ satisfying $\sigma_{u_{\star}}=\xi$ and $\sigma_{v_{\star}}=\zeta$ and $\mu_e$ be the uniform measure on $\Omega_e$. For each $e=\{u,v\}\in E'$, let $A_{e}=\{\sigma_u\neq\sigma_v\}$ be the event for whether the two endpoints of $G_e$ have distinct colors. It holds that $\mu=\mu_e(\cdot \mid A_e)$ for each $e\in E'$. Furthermore, because paths are chordal graphs, \Cref{lem:ACSE-pinned-chordal} and the definition of ACSE (\Cref{def:color-entropy-subadditivity}) imply that for every $e\in E'$ and every $F:\Omega_{e}\to \mathbb{R}^{\geq 0}$,
\begin{equation}\label{eq:path-entropy-conservation}
    \sum_{c\in [q]}\E[\mu_e]{\Ent[\mu_e]{F\mid I_c}}\geq (q-2)\Ent[\mu_e]{F}.
\end{equation}
By the law of total entropy, \eqref{eq:path-entropy-conservation} is equivalent to:
\begin{equation}\label{eq:path-bounded-entropy-loss}
   \sum_{c\in [q]}\Ent[\mu_e]{\E[\mu_e]{F\mid I_c}}\leq 2\Ent[\mu_e]{F}.
\end{equation}

Observe that $\Omega\subseteq \Omega_e$ for each $e\in E'$. We extend the function $f$ on $\Omega$ to a function $F_e\defeq 1+\one{A_e}(f-1)$ on the larger state space $\Omega_e$. Intuitively, this assigns the mean value of $1$ to any invalid cycle configuration, ensuring these extra states contribute exactly zero to the entropy. Formally:
\begin{equation}
    F_e(\sigma)\defeq \begin{cases} f(\sigma) & A_e(\sigma)=1, \\ 1 & A_e(\sigma)=0. \end{cases}
\end{equation}
Because $\E[\mu]{f}=1$, we preserve the mean $\E[\mu_e]{F_e}=1$ for each $e\in E'$, yielding
\begin{equation}\label{eq:entropy-transform}
    \Ent[\mu_e]{F_e}=\mu_e(A_e)\Ent[\mu]{f}.
\end{equation}

For each edge $e\in E'$ and each color $c\in [q]$, define the weight $w_{e,c}(I_c)\defeq \mu_e(A_e\mid I_c)$. Then we can write:
\begin{equation}\label{eq:path-case-decomposition}
\E[\mu_e]{F_e\mid I_c}=1+w_{e,c}(I_c)\tp{\E[\mu]{f\mid I_c}-1}.
\end{equation}
Let $\Phi(t)\defeq t\log t-t+1$ be the entropy function, with the continuous extension $\Phi(0)=1$. Under this definition, $\Ent[\nu]{g}=\E[\nu]{\Phi(g)}$ for any distribution $\nu$ and any $g$ satisfying $\E[\nu]{g}=1$. We also introduce a helper function $\psi:[0,1]\to [0,1]$, formally defined as:
\begin{equation}\label{eq:definition-helper}
\psi(t)\defeq \begin{cases}0 & t=0, \\\frac{t+(1-t)\log (1-t)}{t} & 0<t<1, \\ 1 & t=1. \end{cases}
\end{equation}
Notice that $t\psi(t)=\Phi(1-t)$ for all $t\in [0,1]$.

We make the following claim, which will be verified after the end of the proof.
\begin{claim}\label{claim:functional}
For any $0\leq t\leq 1$ and $x\geq 0$,
\[
    \Phi(1+t(x-1))\geq \Phi(1-t)\Phi(x).
\]
\end{claim}   Combining \eqref{eq:path-case-decomposition} and \Cref{claim:functional}, we obtain for any $e\in E'$ and $c\in [q]$,
\[
\Phi(\E[\mu_e]{F_e\mid I_c})\geq w_{e,c}\psi(w_{e,c})\Phi(\E[\mu]{f\mid I_c}).
\]
Further combining  with \eqref{eq:entropy-transform} gives
\begin{equation}\label{eq:entropy-help-function}
    \Ent[\mu_e]{\E[\mu_e]{F_e\mid I_c}}\geq \mu_e(A_e)\E[\mu]{\psi(w_{e,c}(I_c))\Phi(\E[\mu]{f\mid I_c})}.
\end{equation}
Applying the path entropy bound \eqref{eq:path-bounded-entropy-loss} together with \eqref{eq:entropy-transform}, we find that for any $e\in E'$,
\begin{equation}\label{eq:tilted-entropy-change}
\sum_{c\in [q]}\E[\mu]{\psi(w_{e,c}(I_c))\Phi(\E[\mu]{f\mid I_c})}\leq 2\Ent[\mu]{f}.
\end{equation}
Define $\bar{\psi}_c(I_c)\defeq \frac{1}{n-1}\sum_{e\in E'}\psi(w_{e,c}(I_c))$ as the average over all unpinned edges. Averaging \eqref{eq:tilted-entropy-change} yields
\begin{equation}\label{eq:tilted-entropy-change-average}
\sum_{c\in [q]}\E[\mu]{\bar{\psi}_c(I_c)\Phi(\E[\mu]{f\mid I_c})}\leq 2\Ent[\mu]{f}.
\end{equation}

Note that \eqref{eq:tilted-entropy-change-average} isolates the entropy loss bounded by the path graph comparison. We can decompose the total entropy loss as follows:
\begin{equation}\label{eq:entropy-loss-separation}
\begin{aligned}
L(f)=&\sum_{c\in [q]}\Ent[\mu]{\E[\mu]{f\mid I_c}}\\
=&\sum_{c\in [q]}\E[\mu]{\Phi(\E[\mu]{f\mid I_c})}\\
=&\underbrace{\sum_{c\in [q]}\E[\mu]{\bar{\psi}_c(I_c)\Phi(\E[\mu]{f\mid I_c})}}_{\text{path comparison term}} + \underbrace{\sum_{c\in [q]}\E[\mu]{(1-\bar{\psi}_c(I_c))\Phi(\E[\mu]{f\mid I_c})}}_{\text{residual term}}.
\end{aligned}
\end{equation}

It remains to bound the residual term by analyzing $\bar{\psi}_c(I_c)$. Let 
    \[
    \alpha_c(I_c)\defeq 2|I_c|-\one{c=\xi\text{ or }c=\zeta} 
    \]
    be the number of unpinned edges incident to $I_c$. We make the following claim, which will also be verified after the end of the proof.

\begin{claim}\label{claim:estimate}
It holds that
\[
1-\bar{\psi}_c(I_c)\leq \tp{1-\frac{\alpha_c(I_c)}{n-1}}\log 2.
\]
\end{claim}

Note that Jensen's inequality guarantees $\Phi(\E[\mu]{f\mid I_c})\leq \E[\mu]{\Phi(f)\mid I_c}$. Note that we also have $\sum_{c\in [q]}\alpha_c(I_c)=2(n-1)$, we have by \Cref{claim:estimate} that
\begin{equation}\label{eq:main-term-upper-bound}
    \sum_{c\in [q]}\E[\mu]{\tp{1-\frac{\alpha_c(I_c)}{n-1}}\Phi(\E[\mu]{f\mid I_c}) }\leq \E[\mu]{\Phi(f)\sum_{c\in [q]}\tp{1-\frac{\alpha_c(I_c)}{n-1}} }= (q-2)\Ent[\mu]{f}.
\end{equation}

Finally, substituting \eqref{eq:tilted-entropy-change-average} and \eqref{eq:main-term-upper-bound} into the decomposition \eqref{eq:entropy-loss-separation} gives the upper bound on the entropy loss:
\[
\sum_{c\in [q]}\Ent[\mu]{\E[\mu]{{f\mid I_c}}}\leq \tp{2+(q-2)\log 2 }\Ent[\mu]{f},
\]
and applying the law of total entropy gives the lower bound on the entropy conservation:
\[
\sum_{c\in [q]}\E[\mu]{\Ent[\mu]{{f\mid I_c}}}\geq (q-2)\tp{1-\log 2}\Ent[\mu]{f}.
\]
This immediately concludes the proof of \Cref{lem:ACSE-cycle}.

To prove \Cref{lem:refined-ACSE-cycle}, we claim a sharper estimate that for every $c\in [q]$,
\begin{equation}\label{eq:cycle-sharper-estimate}
    (1 - \bar\psi_c(I_c)) \Phi(\E{f \mid I_c})+(\log 2)\cdot\one{I_c=\emptyset}\cdot\Ent{f\mid I_c}\leq (\log 2)\tp{1-\frac{\alpha_c(I_c)}{n-1}}\E{\Phi(f)\mid I_c}.
\end{equation}
We then verify \eqref{eq:cycle-sharper-estimate}.  For the case when $I_c=\emptyset$, we have
\begin{align*}
&(1 - \bar\psi_c(I_c)) \Phi(\E{f \mid I_c})+(\log 2)\Ent{f\mid I_c}\\
(\text{by \Cref{claim:estimate}})\quad\leq &(\log 2)\tp{\Phi(\E{f \mid I_c})+\E{\Ent{f\mid I_c}}}\\
(\text{law of total entropy})\quad =& (\log 2)\E{\Phi(f)\mid I_c}.
\end{align*}
For the case when $I_c\neq \emptyset$, we have
\begin{align*}
&(1 - \bar\psi_c(I_c)) \Phi(\E{f \mid I_c})\\
(\text{by \Cref{claim:estimate}})\quad\leq &(\log 2)\tp{1-\frac{\alpha_c(I_c)}{n-1}}\Phi(\E{f \mid I_c})\\
(\text{law of total entropy})\quad \leq& (\log 2)\tp{1-\frac{\alpha_c(I_c)}{n-1}}\E{\Phi(f)\mid I_c},
\end{align*}
therefore \eqref{eq:cycle-sharper-estimate} holds in both cases.

Combining \eqref{eq:tilted-entropy-change-average}, \eqref{eq:entropy-loss-separation} and \eqref{eq:cycle-sharper-estimate}, we have that
\begin{align*}
L(f)+(\log 2)H_\emptyset(f)\leq &2\Ent{f}+\log 2\sum\limits_{c\in [q]}\E{\tp{1-\frac{\alpha_c(I_c)}{n-1}}\E{\Phi(f)\mid I_c}}\\
=&2\Ent{f}+\log 2\E{\Phi(f)\sum\limits_{c\in [q]}\tp{1-\frac{\alpha_c(I_c)}{n-1}}}\\
=&\tp{2+(q-2)\log 2}\Ent{f}.
\end{align*}
Combining with \eqref{eq:entropy-loss}, we obtain
\begin{align*}
H_{+}(f)+(1-\log 2)H_{\emptyset}(f)=&H_{+}(f)+H_{\emptyset}(f)-(\log 2)H_{\emptyset}(f)\\
=&q\Ent[\mu]{f}-L(f)-(\log 2)H_{\emptyset}(f)\\
\geq &\tp{q-2-(q-2)\log 2}\Ent[\mu]{f}\\
= & (1-\log 2)(q-2)\Ent[\mu]{f}.
\end{align*}
This gives \eqref{eq:refined-entropy-conservation-cycle-constant}, hence completing the proof of \Cref{lem:refined-ACSE-cycle}.
\end{proof}

\begin{proof}[Proof of \Cref{claim:functional}]
 The boundary cases $t=0$ and $t=1$ are trivial. For $t \in (0,1)$, define $g(x)\defeq\Phi(1+t(x-1))- \Phi(1-t)\Phi(x)$. Differentiating with respect to $x$ yields
\[
g'(x)=t\log(1-t+tx)-\Phi(1-t)\log x,
\]
and a second derivative of
\begin{equation}\label{eq:second-derivative}
g''(x)=\frac{t^2}{1-t+tx}-\frac{\Phi(1-t)}{x}=\frac{t(t-\Phi(1-t))x-\Phi(1-t)(1-t)}{x(1-t+tx)}.
\end{equation}
The denominator of \eqref{eq:second-derivative} is strictly positive for $x > 0$. The numerator is an affine linear function of $x$ with a negative intercept ($-\Phi(1-t)(1-t)<0$). Note that we have:
\[
t^2-\Phi(1-t)=(1-t)\tp{(1-t)-1-\log{(1-t)}}>0.
\]
Because $\Phi(1-t)<t^2< t$, the slope is strictly positive. Consequently, $g''(x)$ crosses zero at exactly one point $x_0 > 0$. Evaluating the numerator at $x=1$ gives $t^2-\Phi(1-t)>0$, guaranteeing this root of $g''(x)$ lies in $x_0\in (0,1)$. Therefore $g'(x)$ decreases in $(0,x_0)$ and increases in $(x_0,+\infty)$. Note that $\lim\limits_{x\to 0^{+}}g'(x)=+\infty$ and $g'(1)=0$, therefore we have $g'(x_0)<0$ and that $g'(x)$ has only one zero before $x_0$, becomes negative before reaching $1$, and turns positive again for $x>1$. Because $g(0)=g(1)=0$, the function $g(x)$ must be globally non-negative on $x \ge 0$. This completes the proof of \Cref{claim:functional}.
\end{proof}

\begin{proof}[Proof of \Cref{claim:estimate}]
    
We split into two cases, the empty color case ($I_c=\emptyset$) and the nonempty color case ($I_c\neq \emptyset$).

If $I_c\neq \emptyset$, then $I_c$ is an independent set of $G$, and exactly $2\abs{I_c}$ edges (possibly including the pinned edge $e_{\star}$) are incident to $I_c$. If the removed edge $e\in E'$ is incident to $I_c$, then $A_e$ is uniquely determined by $I_c$, giving $w_{e,c}(I_c)=1$. Conversely, suppose $e$ is not incident to $I_c$. The endpoints of $e$ are the two endpoints of the path $G_e$. Removing the nonempty set $I_c$ separates them into different components. Specifically,
\begin{itemize}
    \item if $c\in [q]\setminus \{\xi,\zeta\}$, both pinned vertices survive and remain in the same component because the edge $e_{\star}$ is present;
    \item if $c=\xi$ or $c=\zeta$, one pinned vertex is removed and only the other may survive.
\end{itemize}
Thus at least one endpoint of $e$ belongs to a component with no pinning. Conditional on $I_c$,  the components are independent, and an unpinned component is a uniformly colored path with color set $[q]\setminus \{c\}$. Its endpoint is therefore uniform on those $q-1$ colors. Consequently, we have $w_{e,c}(I_c)=\frac{q-2}{q-1}\geq \frac{1}{2}$.
    
    Differentiating gives $\psi'(t)=\frac{-\log(1-t)-t}{t^2}>0$, therefore $\psi(t)$ is strictly increasing on $[0,1)$, hence we have $\psi(w_{e,c})\geq \psi\tp{\frac{1}{2}}=1-\log 2$. Consequently we have for each $c\in [q]$:
\begin{equation}\label{eq:psi-calculation-nonempty}
\bar{\psi}_c(I_c)\geq \frac{\alpha_c(I_c)}{n-1}\psi(1)+\tp{1-\frac{\alpha_c(I_c)}{n-1}}\psi\tp{\frac{1}{2}}=1- \tp{1-\frac{\alpha_c(I_c)}{n-1}}\log 2 , \quad I_c\neq \emptyset.
\end{equation}

Now consider the case $I_c=\emptyset$. This event has positive probability only if $c\in [q]\setminus \{\xi,\zeta\}$. Under $\mu_e(\cdot\mid I_c=\emptyset)$, the deleted-edge path is distributed as a uniformly random proper coloring of a path with color set $[q]\setminus \{c\}$ subject to the pinning on $e_{\star}$. Note that to calculate $w_{e,c}$, by color symmetry we can ignore the pinning on $e_{\star}$. In this case, each $w_{e,c}$ equals the probability that the two endpoints of a length-$(n-1)$ proper $(q-1)$-coloring path are distinct, which is:
\begin{equation}\label{eq:psi-calculation-empty}
\bar{\psi}_c(\emptyset)=\psi(h_{q,n}),\quad \text{where }h_{q,n}\defeq\tp{\frac{q-2}{q-1}\tp{1-\tp{-\frac{1}{q-2}}^{n-1}}}.
\end{equation}

Note that when $q=3$, we have $h_{3,n}=\one{n \text{ even}}$, therefore positive probability of $I_c=\emptyset$ requires $n$ even, and then $h_{3,n}=1$. When $q\geq 4$, we have when $n$ is even:
\[
h_{q,n}=\frac{q-2}{q-1}\tp{1+(q-2)^{-(n-1)}}\geq \frac{q-2}{q-1}>\frac{1}{2}.
\]
And when $n$ is odd:
\[
h_{q,n}=\frac{q-2}{q-1}\tp{1-(q-2)^{-(n-1)}}\geq \frac{q-2}{q-1}\tp{1-(q-2)^{-2}}=\frac{q-3}{q-2}\geq\frac{1}{2}.
\]

 Therefore by comparing the lower bounds in \eqref{eq:psi-calculation-nonempty} and \eqref{eq:psi-calculation-empty}, we obtain the proof of \Cref{claim:estimate}.
\end{proof}

We mentioned that the pinned-cycle lemma (\Cref{lem:ACSE-cycle}) has a worse ACSE constant compared to chordal graphs and that directly plugging \Cref{lem:ACSE-cycle} into \Cref{thm:subadditivity-implies-tensorization} would lead to a super-polynomial dependency with respect to $q$ in the ACTE constant and the mixing time of the WSK dynamics. However, note that after revealing one nonempty color class, the pinned cycle breaks down into (pinned) chains, this allows us to prove the following lemma that directly bounds the ACTE constant through the strengthened \Cref{lem:refined-ACSE-cycle}. Note that this can be viewed as a certain kind of average-case local-to-global argument in spirit~\cite{anari2022optimal}.

\begin{lemma}[ACTE on a pinned cycle]\label{lem:q-minus-two-ACSE-cycle}
 Let $G=(V,E)$ be a cycle graph with $|V|\geq 3$ and let $q\geq 3$. Fix an edge $e_{\star}=\{u_{\star},v_{\star}\}\in E$ and distinct colors $\xi,\zeta\in [q]$. Let $\mu$ be the uniform measure on proper colorings $\sigma$ of $G$ satisfying $\sigma_{u_{\star}}=\xi$ and $\sigma_{v_{\star}}=\zeta$. Then $\mu$ satisfies approximate colorwise tensorization of entropy with constant $1/(1-\log 2)$. 
\end{lemma}

\begin{proof}
From \Cref{def:color-entropy-tensorization}, we need to prove for any $f: \Omega \to \mathbb{R}^{\ge 0}$ that
\begin{equation}\label{eq:entropy-conservation-cycle-all}
    \Ent[\mu]{f} \le \frac{1}{1-\log 2} \sum\limits_{S \in \binom{[q]}{q-2}} \E[\mu]{\Ent[\mu]{f \mid I_{S}}}.
\end{equation}

We now proceed to prove \eqref{eq:entropy-conservation-cycle-all} by induction on $q$. The base case $q=3$ directly follows from \Cref{lem:ACSE-cycle}.

For the inductive step $q \ge 4$. We recall the notations from \Cref{lem:refined-ACSE-cycle} that
\begin{align*}
    H_{\emptyset}(f)\defeq \sum\limits_{c\in [q]}\E[\mu]{\one{I_c=\emptyset}\Ent[\mu]{{f\mid I_c}}},\quad H_{+}(f)\defeq \sum\limits_{c\in [q]}\E[\mu]{\one{I_c\neq \emptyset}\Ent[\mu]{{f\mid I_c}}}.
\end{align*}
We additionally define the following quantities:
\[
    K_\emptyset(f) \defeq \sum\limits_{c\in [q]}\E[\mu]{\one{I_c=\emptyset}\sum\limits_{R\in\binom{[q]\setminus \{c\}}{q-3}}\Ent[\mu]{f\mid I_{R\cup \{c\}}}},
\]
\[
   K_+(f) \defeq \sum\limits_{c\in [q]}\E[\mu]{\one{I_c\neq\emptyset}\sum\limits_{R\in\binom{[q]\setminus \{c\}}{q-3}}\Ent[\mu]{f\mid I_{R\cup \{c\}}}}.
\]
Note that each fixed $S\in \binom{[q]}{q-2}$ appears exactly $q-2$ times as $\{c\}\cup R$, once for each $c\in S$, therefore we have
\begin{equation}\label{eq:counting-identity}
    K_\emptyset(f)+K_+(f)=(q-2)\sum\limits_{S\in \binom{[q]}{q-2}}\E[\mu]{\Ent[\mu]{f\mid I_S}}.
\end{equation}

We now lower bound $K_\emptyset(f)$ and $K_+(f)$.

First, consider the case $I_c=\emptyset$. On such event, we have $c$ cannot be $\xi$ or $\zeta$. The conditional measure is exactly the pinned-cycle coloring measure with color set $[q]\setminus \{c\}$, hence with $q-1$ colors. Applying the induction hypothesis gives
\[
\sum\limits_{R\in \binom{[q]\setminus \{c\}}{q-3}}\E[\mu]{\Ent[\mu]{f\mid I_c,I_R}\mid I_c}\geq (1-\log 2)\Ent[\mu]{f\mid I_c}
\]
on the event $I_c=\emptyset$. Summing over all available $c$ gives
\begin{equation}\label{eq:empty-branch}
    K_\emptyset(f) \ge (1-\log 2)H_\emptyset(f).
\end{equation}

Then, consider the case $I_c\neq \emptyset$. After conditioning on such a color class and deleting those vertices, the cycle breaks into a disjoint union of paths. If both pinned vertices remain, they are still joined through the pinned edge $e^{\star}$; if one pinned color equals $c$, then only one pinned vertex remains. Hence the surviving pinned vertices form a clique of size at most two. 

Therefore the conditional measure on the remaining vertices is a pinned chordal coloring measure with $q-1$ available colors. Iterating the pinned-chordal ACSE lemma (\Cref{lem:ACSE-pinned-chordal}), or equivalently using its ACTE consequence with constant $1$ gives
\[
\sum\limits_{R\in \binom{[q]\setminus \{c\}}{q-3}}\E[\mu]{\Ent[\mu]{f\mid I_c,I_R}\mid I_c}\geq \Ent[\mu]{f\mid I_c}
\]
on the event $I_c\neq \emptyset$. Summing over all $c$ gives
\begin{equation}\label{eq:nonempty-branch}
    K_+(f) \ge H_+(f).
\end{equation}

Combining \eqref{eq:empty-branch}, \eqref{eq:nonempty-branch} and \Cref{lem:refined-ACSE-cycle} we have:
\[
      K_{\emptyset}(f)+K_{+}(f)\geq (1-\log 2)(q-2)\Ent[\mu]{f}.
\]

Further combining with \eqref{eq:counting-identity} and normalizing gives \eqref{eq:entropy-conservation-cycle-all}, concluding the proof of the lemma.
\end{proof}

Using the ACTE on pinned cycles established in \Cref{lem:q-minus-two-ACSE-cycle}, we are able to apply an inductive method to establish ACTE on outerplanar graphs, therefore proving \Cref{lem:ACTE-outerplanar}.

\begin{proof}[Proof of \Cref{lem:ACTE-outerplanar}]

      We construct a sequence 
       \[
    \emptyset = G_0 \subset G_1 \subset \dots \subset G_m = G
    \]
    according to \Cref{cor:structural-decomposition}. We claim that ACTE with constant $1/(1-\log 2)$ holds for all $\mu_i={\mu_{G_i,q}}$ satisfying $1\leq i\leq m$.
     
    \textbf{Base Case ($i=1$):} By \Cref{cor:structural-decomposition} and that $G_0=\emptyset$, $G_1$ can only be a singleton, which is chordal. Iterating chordal ACSE in \Cref{lem:ACSE-chordal} through $q-2$ color reveals yields
$$\frac{1}{\binom{q}{2}}\sum\limits_{\{a,b\}\in \binom{[q]}{2}}\E{\Ent{f\mid I_{[q]\setminus \{a,b\}}}}\geq \tp{\prod\limits_{r=3}^{q}\frac{r-2}{r}}\Ent{f}=\frac{2}{q(q-1)}\Ent{f}.$$
Multiplying both sides by $\binom{q}{2} = \frac{q(q-1)}{2}$, we obtain $\sum_{S \in \binom{[q]}{q-2}} \E{\Ent{f \mid I_S}} \geq \Ent{f}$, therefore the lemma holds for the base case.

\textbf{Inductive Step ($i > 1$):} Let $G_i=G_{i-1}\cup J_i$ and $H_i=G_{i-1}\cap J_i$. By \Cref{cor:structural-decomposition}, one of the following two cases must hold:

\textbf{Case I} (\Cref{cor:structural-decomposition}-(\ref{item:outplanar-1}),(\ref{item:outplanar-2})):  $J_i$ is a singleton vertex and $H_i=\emptyset$, or $J_i$ is a single edge and $H_i$ consists of a single vertex. In this case, $V(G_i)\setminus V(G_{i-1})=\{v\}$ for some simplicial vertex $v\in V(G_i)$. 

Every proper coloring of $G_{i-1}$ can be extended to a proper coloring of $G_i$ in exactly the same number of ways. Therefore, the distribution $\mu_{G_{i-1},q}$ is exactly the marginal distribution of $\mu$ on $V(G_i)\setminus \{v\}$. 

For each $\sigma \in \Omega_{G_i,q}$ and each subset of colors $S \in \binom{[q]}{q-2}$, the information for revealing all vertices colored $S$ is $I_S(\sigma) = \tp{u \in V(G_i) \mid \sigma_u =c }_{c\in S}$. We define the partial information: 
$$I_S^{0}(\sigma) \defeq \tp{\{u \in V(G_i) \setminus \{v\} \mid \sigma_u =c\} }_{c\in S}, \quad B_S(\sigma) \defeq \tp{ \one{\sigma_v =c} }_{c\in S}.$$
The pair $(I_S^0, B_S)$ encodes the exact same information as $I_S$. Applying the law of total entropy, we have for each $S\in \binom{[q]}{q-2}$:
\begin{equation}\label{eq:chordal-induction-law-of-total-entropy-2}
    \E{\Ent{f\mid I_S}} = \E{\Ent{\E{f\mid \sigma_{V\setminus \{v\}},B_S}\mid I_S}} + \E{\Ent{f\mid \sigma_{V\setminus \{v\}},B_S}},
\end{equation}
where in the final term we dropped the conditioning on $I_S$ because the partial configuration $\sigma_{V\setminus \{v\}}$ together with $B_S$ strictly determines $I_S$.

We bound the sum over all $S \in \binom{[q]}{q-2}$ for the two terms on the RHS of \eqref{eq:chordal-induction-law-of-total-entropy-2} separately. Let $\bar{f}\defeq \E{f\mid \sigma_{V\setminus \{v\}}}$ be the function averaged over $v$. For the first term, note that:
$$\E{\Ent{\E{f\mid \sigma_{V\setminus \{v\}},B_S}\mid I_S}} = \E{\Ent{\E{f\mid \sigma_{V\setminus \{v\}},B_S}\mid I_{S}^0,B_S}} \geq \E{\Ent{\bar{f}\mid I_S^0}},$$
where the inequality uses the convexity of entropy and the fact that $B_S$ is independent of $\sigma_{V\setminus \{v\}}$ conditioned on $I^0_S$. Summing over all subsets $S$ and applying the induction hypothesis to $G_{i-1}$, we have:
\begin{equation}\label{eq:chordal-induction-first-part-2}
    \sum\limits_{S\in \binom{[q]}{q-2}} \E{\Ent{\E{f\mid \sigma_{V\setminus \{v\}},B_S}\mid I_S}} \geq \E{\sum\limits_{S\in \binom{[q]}{q-2}}\Ent{\bar{f}\mid I_S^0}} \geq (1 - \log 2)\Ent{\bar{f}}.
\end{equation}

We now focus on the second term. Because $v$ is simplicial, conditioning on $\sigma_{V\setminus \{v\}}$ leaves $\sigma_v$ uniformly distributed over a set of available colors $T$, where $|T| \geq 1$.   we are guaranteed $1\leq |T|\leq q$.  For a set $S\in \binom{[q]}{q-2}$ of $q-2$ colors, there are additionally two cases:
\begin{itemize}
    \item $\abs{T\setminus S}\leq 1$, in this case $f$ is completely determined by $\sigma_{V\setminus \{v\}}$ and $B_S$, hence in this case we have
    \[
    \E{\Ent{f\mid \sigma_{V\setminus \{v\}}, B_S}\mid \sigma_{V\setminus \{v\}}}=0.
    \]
    \item $\abs{T\setminus S}=2$. In this case we let $|T|=s$. We first assume $s>2$. Note that the ACSE bound for a singleton (\Cref{lem:bounded-entropy-decay-singleton}) gives an ACSE constant smaller than $\frac{1}{s-2}$, therefore a telescopic product over $s-2$ colors yields an ACTE constant of at most $1$. This gives 
    \begin{equation}\label{eq:outerplanar-induction-case-2}
        \sum\limits_{R\in \binom{T}{s-2}}\E{\Ent{f\mid \sigma_{V\setminus \{v\}}, B_R}\mid \sigma_{V\setminus \{v\}}} \geq \Ent{f\mid \sigma_{V\setminus \{v\}}}.
    \end{equation}
    Also note that \eqref{eq:outerplanar-induction-case-2} follows trivially when $s=2$.
\end{itemize}
   Combining the two cases above, we obtain that
\begin{equation}\label{eq:chordal-induction-averaging-2}
    \sum\limits_{S\in \binom{[q]}{q-2}}\E{\Ent{f\mid \sigma_{V\setminus \{v\}}, B_S}\mid \sigma_{V\setminus \{v\}}} \geq \Ent{f\mid \sigma_{V\setminus \{v\}}}.
\end{equation}
Averaging over all configurations $\sigma_{V\setminus \{v\}}$ and noting that $1 \geq 1 - \log 2$, we obtain:
\begin{equation}\label{eq:chordal-induction-second-part-2}
    \sum\limits_{S\in \binom{[q]}{q-2}}\E{\Ent{f\mid \sigma_{V\setminus \{v\}}, B_S}} \geq (1 - \log 2)\E{\Ent{f\mid \sigma_{V\setminus \{v\}}}}.
\end{equation}

Combining \eqref{eq:chordal-induction-law-of-total-entropy-2}, \eqref{eq:chordal-induction-first-part-2}, and \eqref{eq:chordal-induction-second-part-2} yields:
$$\sum\limits_{S\in \binom{[q]}{q-2}}\E{\Ent{f\mid I_S}} \geq (1 - \log 2)\Ent{\bar{f}} + (1 - \log 2)\E{\Ent{f\mid \sigma_{V\setminus \{v\}}}} = (1 - \log 2)\Ent{f},$$
where the last equality follows directly from the law of total entropy. This concludes Case I.

\textbf{Case II} (\Cref{cor:structural-decomposition}-(\ref{item:outplanar-3})): $J_i$ is a cycle $(v_1, v_2, \dots, v_k)$ with $k \geq 3$ such that $\mathrm{deg}(v_i) = 2$ for all $1 < i < k$, and $H_i$ consists of the single edge $\{v_1, v_k\}$.
  Let $W=\{v_i\}_{i=2}^{k-1}$. The induced subgraph $G_{i-1} = G_i \setminus W$ is also outerplanar. Every proper coloring of $G_{i-1}$ can be extended to a proper coloring of $G_i$ in exactly the same number of ways. Therefore, the distribution $\mu_{G_{i-1},q}$ is exactly the marginal distribution of $\mu$ on $V(G_i) \setminus W$.

Similarly to Case I, we define the partial information: 
$$I_S^{0}(\sigma) \defeq \tp{\{u \in V(G_i) \setminus W \mid \sigma_u =c \} }_{c\in S}, \quad B_S(\sigma) \defeq \tp{u \in W \mid \sigma_u =c}_{c\in S}.$$
Applying the law of total entropy, we have for each $S\in \binom{[q]}{q-2}$:
\begin{equation}\label{eq:outerplanar-induction-law-of-total-entropy}
     \E{\Ent{f\mid I_S}} = \E{\Ent{\E{f\mid \sigma_{V\setminus W},B_S}\mid I_S}} + \E{\Ent{f\mid \sigma_{V\setminus W},B_S}},
\end{equation}
where the conditioning on $I_S$ is dropped in the final term because $\sigma_{V\setminus W}$ and $B_S$ strictly determine $I_S$.

Let $\bar{f}\defeq \E{f\mid \sigma_{V\setminus W}}$. Summing the first term over all choices of $S$ and applying the induction hypothesis to $G_{i-1}$ yields:
\begin{equation}\label{eq:outerplanar-induction-first-part}
    \sum\limits_{S\in \binom{[q]}{q-2}} \E{\Ent{\E{f\mid \sigma_{V\setminus W},B_S}\mid I_S}} \geq \E{\sum\limits_{S\in \binom{[q]}{q-2}}\Ent{\bar{f}\mid I_S^0}} \geq (1-\log 2)\Ent{\bar{f}}.
\end{equation}

For the second term, because $\{v_1,v_k\}\in E(G_i)$, any valid coloring $\sigma\in \Omega_{G_i,q}$ must assign distinct colors to $v_1$ and $v_k$. Applying the ACTE for pinned cycles (\Cref{lem:q-minus-two-ACSE-cycle}) gives:
\begin{equation}\label{eq:outerplanar-induction-averaging}
    \sum\limits_{S\in \binom{[q]}{q-2}}\E{\Ent{f\mid \sigma_{V\setminus W}, B_S}\mid \sigma_{V\setminus W}} \geq (1-\log 2)\Ent{f\mid \sigma_{V\setminus W}}.
\end{equation}
Averaging over all configurations $\sigma_{V\setminus W}$, we obtain:
\begin{equation}\label{eq:outerplanar-induction-second-part}
    \sum\limits_{S\in \binom{[q]}{q-2}}\E{\Ent{f\mid \sigma_{V\setminus W}, B_S}} \geq (1-\log 2)\E{\Ent{f\mid \sigma_{V\setminus W}}}.
\end{equation}

Finally, combining \eqref{eq:outerplanar-induction-law-of-total-entropy}, \eqref{eq:outerplanar-induction-first-part}, and \eqref{eq:outerplanar-induction-second-part} yields:
$$\sum\limits_{S\in \binom{[q]}{q-2}}\E{\Ent{f\mid I_S}} \geq (1-\log 2)\Ent{\bar{f}} + (1-\log 2)\E{\Ent{f\mid \sigma_{V\setminus W}}} = (1-\log 2)\Ent{f}.$$
This concludes the inductive step and the proof.
\end{proof}

\section{A Lower Bound on the Mixing Time of the WSK Dynamics}\label{sec:lower-bound}

In this section we prove \Cref{thm:WSK-rapid-mixing-lower}, a lower bound on the mixing time of the WSK dynamics. Our proof is based on a simple distinguishing-statistic argument. We need the following standard lemma for establishing lower bounds for Markov chain mixing times. 

\begin{lemma}[{\cite[Proposition 7.8]{levin2017markov}}]\label{lem:distinguishing-statistics}
    Let $\mu$ and $\nu$ be two probability distributions on $\Omega$, and let $f$
be a real-valued function on $\Omega$. If
\[
\abs{\E[\mu]{f}-\E[\nu]{f}}\geq r\sigma
\]
where $\sigma^2=\frac{\Var[\mu]{f}+\Var[\nu]{f}}{2}$, then
\[
\DTV{\mu}{\nu}\geq 1-\frac{4}{4+r^2}.
\]
\end{lemma}

We are now ready to prove the lower bound for the WSK dynamics. The distinguishing statistic we choose is simple: we track the number of pairs of adjacent vertices with a particular color pair, and estimate how it deviates from its mean.

\begin{proof}[Proof of \Cref{thm:WSK-rapid-mixing-lower}]
Let $G=(V,E)$ where $E=\{\{i,i+1\}\mid 1\leq i<n\}$. Let $\Omega=\Omega_{G,q}$ and $\mu=\mu_{G,q}$ and $P$ denote the transition matrix of the WSK dynamics. Set the starting configuration $X_0$ of the WSK dynamics as 
\[
\sigma_0=(1,2,1,2,\dots).
\]
We will show that there exists a constant $C=C(q)$ such that for sufficiently large $n$ and
\[
t=\left\lfloor \frac{1}{2} \log_2 (n-1)-\frac{1}{2}\log_2\log_2(n-1)-C \right\rfloor
\]
it holds that 
\begin{equation}\label{eq:total-variation-distance-lb}
\DTV{P^t(\sigma_0,\cdot)}{\mu}>\frac{1}{4},
\end{equation}
which proves the theorem.

The distinguishing statistic we choose depends on the fraction of edges with a particular pair of colors. Formally, for each unordered pair $\{a,b\}\in \binom{[q]}{2}$, define
\[
Y_{\{a,b\}}(\sigma)\defeq\frac{1}{n-1}\left|\set{1\leq i\leq n-1\mid \{\sigma_i,\sigma_{i+1}\}=\{a,b\} }\right|, \quad  \forall \sigma \in \Omega.
\]
For the stationary distribution $\mu$, by color symmetry we have $\E[\mu]{Y_{\{a,b\}}}=\frac{1}{\binom{q}{2}}$ for each $\{a,b\}\in \binom{[q]}{2}$.

We then analyze the evolution of $Y_{\{a,b\}}$ in the WSK dynamics starting from $X_0=\sigma_0$. Note that initially we have that $Y_{\{1,2\}}(\sigma_0)=1$ and $Y_{\{a,b\}}(\sigma_0)=0$ for all $\{a,b\}\neq \{1,2\}$. 

For each time $t\geq 0$, let $A_t\in \binom{[q]}{2}$ denote the color pair chosen by the WSK dynamics at time $t$. We also let $\+A_t=(A_1,A_2,\dots,A_t)$ be the sequence of color pairs generated up to time $t$. 
A simple observation by the transition of the WSK dynamics is that $2^t \E{Y_{\{a,b\}}(X_t)\mid \+A_t}\in \=Z$ for each $\{a,b\}\in \binom{[q]}{2}$. While $\binom{q}{2}$ is never a power of two for $q\geq 3$, we have that for each $\{a,b\}\in \binom{[q]}{2}$ and each $t\geq 0$,
\begin{equation}\label{eq:statistics-bias}
\abs{ \E{Y_{\{a,b\}}(X_t)\mid \+A_t}-\frac{1}{\binom{q}{2}}}\geq \delta_t, \quad \text{where } \delta_t\defeq \frac{1}{2^t\binom{q}{2}}.
\end{equation}

We establish the following concentration bound for the distinguishing statistic we chose under the evolution of the WSK dynamics, which will be proved at the end of the section.
\begin{lemma}[Concentration bound for the WSK dynamics]\label{lem:concentration-1}
It holds that
\[
\Var{Y_{\{1,2\}}(X_t)\mid \+A_t}\leq \frac{t}{n-1}.
\]
\end{lemma}
Then by \Cref{lem:concentration-1}, Chebyshev's inequality and combining with \eqref{eq:statistics-bias}, we obtain:
\[
\Pr{\abs{Y_{\{1,2\}}(X_t)-\frac{1}{\binom{q}{2}}}\geq \frac{\delta_t}{2}\mid \+A_t}\geq 1-\frac{4t}{(n-1)\delta_t^2}.
\]
Further averaging over $\+A_t$ gives:
\begin{equation}\label{eq:prob-bound-1}
    \Pr{\abs{Y_{\{1,2\}}(X_t)-\frac{1}{\binom{q}{2}}}\geq \frac{\delta_t}{2}}\geq 1-\frac{4t}{(n-1)\delta_t^2}.
    \end{equation}

Now we give the concentration bound for the distinguishing statistic we chose under the stationary distribution, which will also be proved at the end of the section through a routine calculation.
\begin{lemma}[Concentration bound for the stationary distribution]\label{lem:concentration-2}
    There exists some constant $C_0=C_0(q)$ such that
    \[
    \Var[\mu]{Y_{\{1,2\}}}\leq \frac{C_0}{n-1}.
    \]
\end{lemma}
Again by \Cref{lem:concentration-2} and Chebyshev's inequality we have that
\begin{equation}\label{eq:prob-bound-2}
    \Pr[\mu]{\abs{Y_{\{1,2\}}-\frac{1}{\binom{q}{2}}}\geq \frac{\delta_t}{2}}\leq \frac{4C_0}{(n-1)\delta_t^2}.
    \end{equation}

To apply \Cref{lem:distinguishing-statistics} to prove \eqref{eq:total-variation-distance-lb}, we will choose the function $f$ as follows:
\begin{equation}\label{eq:definition-f}
    E_t\defeq \set{\abs{Y_{\{1,2\}}-\frac{1}{\binom{q}{2}}}\geq \frac{\delta_t}{2}},\quad f\defeq \one{E_t}.
\end{equation}

    Note that we can choose $C$ large enough such that for all sufficiently large $n$,
    \begin{equation}\label{eq:constant-choice}
        \frac{4t}{(n-1)\delta_t^2}\leq \frac{1}{8},\quad \frac{4C_0}{(n-1)\delta_t^2}\leq \frac{1}{8}.
    \end{equation}
    Let $\nu_t$ denote the distribution $P^t(\sigma_0,\cdot)$. Then by combining \eqref{eq:prob-bound-1}, \eqref{eq:prob-bound-2},  \eqref{eq:definition-f} and \eqref{eq:constant-choice} we obtain
    \begin{equation}
    \E[\nu_t]{f}\geq \frac{7}{8}, \quad \E[\mu]{f}\leq \frac{1}{8}, \quad \abs{\E[\nu_t]{f}-\E[\mu]{f}}\geq \frac{3}{4}.    
    \end{equation}

    Since $f$ is an indicator function, we also have 
    \[
    \Var[\nu_t]{f},\Var[\mu]{f}\leq \frac{1}{8}\times\tp{\frac{7}{8}}^2+\frac{7}{8}\times\tp{\frac{1}{8}}^2=\frac{7}{64}.    
    \]

    Therefore plugging into \Cref{lem:distinguishing-statistics} we can derive \eqref{eq:total-variation-distance-lb}, thereby proving the theorem. 
\end{proof}

We now conclude the section by proving \Cref{lem:concentration-1} and \Cref{lem:concentration-2}.
\begin{proof}[Proof of \Cref{lem:concentration-1}]
We claim the following stronger one-step bound that for any $\rho\in \binom{[q]}{2}$ and any $t\geq 1$ it holds that
\begin{equation}\label{eq:vect-concentration-stronger}
\Var{Y_{\rho}(X_t)\mid X_{t-1},\+A_t}\leq \frac{1}{n-1}.
\end{equation}
Before proving \eqref{eq:vect-concentration-stronger}, we show that \eqref{eq:vect-concentration-stronger} already implies the lemma. For each $t\geq 0$, define
\[
\+M_t\defeq \max\limits_{\{a,b\}\in \binom{[q]}{2}} \Var{Y_{\{a,b\}}(X_t)\mid \+A_t}.
\]
We claim that \eqref{eq:vect-concentration-stronger} implies that
\begin{equation}\label{eq:inductive-argument}
\+M_t\leq \+M_{t-1}+ \frac{1}{n-1},
\end{equation}

To prove the claim, by the law of total variance we have that for each $\rho\in \binom{[q]}{2}$, it holds that
\begin{equation}\label{eq:total-law-variance}
\Var{Y_{\rho}(X_t)\mid \+A_t}=\E{\Var{Y_{\rho}(X_t)\mid X_{t-1},\+A_t}\mid \+A_t}+\Var{\E{Y_{\rho}\mid X_{t-1},\+A_t}\mid \+A_t}.
\end{equation}

Now that the first term on the RHS of \eqref{eq:total-law-variance} is at most $\frac{1}{n-1}$ by \eqref{eq:vect-concentration-stronger}. It remains to control the second term. Let the color pair chosen at time $t$ be $A_t=\{a,b\}$. For a fixed previous coloring $X_{t-1}=\sigma$, the conditional expectation of $Y_{\rho}(X_t)$ falls into the following cases:
\begin{itemize}
    \item $\rho=\{a,b\}$ or $\rho\cap \{a,b\}=\emptyset$, then $Y_{\rho}(X_t)$ is unchanged, i.e., 
    \[
    \E{Y_{\rho}\mid X_{t-1}=\sigma,\+A_t}=Y_{\rho}(\sigma), \quad \Var{\E{Y_{\rho}\mid X_{t-1},\+A_t}\mid \+A_t}= \Var{Y_{\rho}\mid \+A_{t-1}};
    \]
    \item $\rho=\{a,c\}$ for some $c\neq b$, then the WSK dynamics at time $t$ swaps boundary edge types $\{a,c\}$ and $\{b,c\}$ with probability $\frac{1}{2}$. Hence
    \[
    \E{Y_{\rho}\mid X_{t-1}=\sigma,\+A_t}=\frac{Y_{\{a,c\}}(\sigma)+Y_{\{b,c\}}(\sigma)}{2},
    \]
    \begin{align*}
     \Var{\E{Y_{\rho}\mid X_{t-1},\+A_t}\mid \+A_t}=& \Var{\frac{Y_{\{a,c\}}+Y_{\{b,c\}}}{2}\mid \+A_{t-1}}\\
     \leq &\max\{\Var{Y_{\{a,c\}}\mid \+A_{t-1}},\Var{Y_{\{b,c\}}\mid \+A_{t-1}}\}.
    \end{align*}
    The case when $\rho=\{b,c\}$ for some $c\neq a$ is analogous.
\end{itemize}
 Therefore we can conclude that for any $\rho\in \binom{[q]}{2}$, 
\[
\Var{\E{Y_{\rho}\mid X_{t-1},\+A_t}\mid \+A_t}\leq \+M_{t-1},
\]
combining with \eqref{eq:vect-concentration-stronger} and \eqref{eq:total-law-variance} proves \eqref{eq:inductive-argument}. Then the  lemma directly follows from \eqref{eq:inductive-argument} by a simple inductive argument.


 It remains to prove \eqref{eq:vect-concentration-stronger}. Fix some $\rho\in \binom{[q]}{2}$. Given $X_{t-1}$ and $A_t$, the vertices whose colors lie in $A_t$ form disjoint path components. Let these components be $C_1,\dots,C_m$. For a component $C_i$, let $\eta_i\in \{0,1\}$ be the indicator whether it is flipped in the WSK move at $t$, then by the WSK transition,
\[
\Pr{\eta_{i}=0}=\Pr{\eta_{i}=1}=\frac{1}{2}, \quad \forall 1\leq i\leq m.
\]
For each $1\leq i\leq m$, let $\Delta_i\in \=R$ be the change in $Y_{\rho}$ caused by flipping $C_i$ alone. Then those independent simultaneous flips add linearly:
\[
Y_{\rho}(X_t)=Y_{\rho}(X_{t-1})+\sum\limits_{1\leq i\leq m}\eta_i\Delta_{i}, \quad \E{Y_{\rho}(X_t)\mid X_{t-1},\+A_t}=Y_{\rho}(X_{t-1})+\frac{1}{2}\sum\limits_{1\leq i\leq m}\Delta_{i}.
\]
Hence we have that
\[
\Var{ Y_{\rho}(X_t) \mid X_{t-1},\+A_t}=\frac{1}{4}\sum\limits_{1\leq i\leq m}\tp{\Delta_{i}}^2.
\]
Note that flipping some component $C$ can only change the type of at most two boundary edges, therefore we have $\abs{\Delta_{i}}\leq \frac{2}{(n-1)}$ for each $1\leq i\leq m$. Also note that we have $m\leq n-1$, and therefore \eqref{eq:vect-concentration-stronger} follows, concluding the proof of the lemma.

\end{proof}

\begin{proof}[Proof of \Cref{lem:concentration-2}]
 For each $1\leq i<n$, let $Z_i\defeq \one{\{\sigma_i,\sigma_{i+1}\}=\{1,2\}}$. We also write $Y_{\{1,2\}}$ simply as $Y$. Note that we have 
 \[
 Y=\frac{1}{n-1}\sum\limits_{i=1}^{n-1}Z_i, \quad \E[\mu]{Y}=p\defeq\frac{1}{\binom{q}{2}}.
 \]
 We need to explore the covariance structure of $Z_i$ for the variance of $Y$. Note that by routine calculation we have for each $1\leq i<n, d\geq 1,i+d<n$:
 \[
 \Pr{Z_{i+d}=1\mid Z_i=1}=\frac{2+(q-2)\tp{-\frac{1}{q-1}}^{d-1}}{q(q-1)}
 \]
 and that
 \[
 \Cov_{\mu}\tp{Z_i,Z_{i+d}}=p(\Pr{Z_{i+d}=1\mid Z_i=1}-p)=\frac{2(q-2)}{q^2(q-1)^2}\tp{-\frac{1}{q-1}}^{d-1}.
 \]
 Therefore we have that
 \begin{align*}
\Var[\mu]{Y}=&\frac{1}{(n-1)^2}\Var[\mu]{\sum\limits_{i=1}^{n-1}Z_i}\\
=&\frac{1}{(n-1)^2}\tp{(n-1)p(1-p)+2\sum\limits_{d=1}^{n-1}(n-1-d)\Cov_{\mu}(Z_i,Z_{i+d})}\\
=&\frac{1}{(n-1)^2}\tp{(n-1)p(1-p)+\frac{4(q-2)}{q^2(q-1)^2}\sum\limits_{d=1}^{n-1}(n-1-d)\tp{-\frac{1}{q-1}}^{d-1}}\\
=& \frac{p(1-p)+\frac{4(q-2)}{q^3(q-1)}}{n-1}-\frac{4(q-2)}{q^4}\frac{1-\tp{-\frac{1}{q-1}}^{n-1}}{(n-1)^2}\\
\leq & \frac{C_0(q)}{n-1},
 \end{align*}
 concluding the proof of the lemma.
\end{proof}

\section{Conclusions and Open Problems}\label{sec:conclusions}

In this paper, we develop new tools for studying the mixing time of the WSK dynamics for uniformly sampling proper $q$-colorings. We introduce two new criteria, approximate colorwise tensorization of entropy (ACTE) and approximate colorwise subadditivity of entropy (ACSE), for controlling the modified log-Sobolev constant of the WSK dynamics. We also develop new inductive approaches for establishing such criteria on specific types of graphs. As concrete applications, we established optimal mixing time bounds for the WSK dynamics on chordal and outerplanar graphs. We leave the following open problems and directions for future work:
 
 \begin{itemize}
 
 \item Our current inductive method for establishing ACTE and ACSE relies heavily on specific structural constraints of the underlying graphs. A natural next step is to relax these requirements to accommodate more general graph classes. Intuitively, this extension requires generalizing the colorwise localization scheme to identify alternative ``entropy-preserving'' substructures beyond simplicial vertices and pinned cycles that can effectively decouple spatial dependencies during the induction step.
 
 \item While the WSK dynamics is ubiquitous in statistical physics and empirically exhibits rapid mixing across diverse topologies, rigorous upper bounds remain sparse. Notably, numerical experiments~\cite{sokal2001personal} support the conjecture that for $q=3$, the WSK dynamics achieves an $O(1)$ mixing time on the $N \times N$ torus for even $N$. Establishing a mixing time upper bound via our new framework for such graphs is a compelling open problem.
 
 \item We have established that the WSK dynamics achieves optimal mixing in regimes where Glauber dynamics becomes non-ergodic (e.g., edge colorings on trees). For edge coloring on general graphs, Glauber dynamics faces a strict irreducibility barrier below $q = 2\Delta$, whereas WSK dynamics remains irreducible for all $q \geq \Delta + 2$~\cite{narboni2023vizing}. Whether we can show rapid mixing within this broader regime for general edge colorings also remains a major open problem.
 
 \end{itemize}

\ifdoubleblind
\else
\section*{Acknowledgements}
We thank Chihao Zhang for helpful discussions. Chunyang Wang thanks Xinyuan Zhang for helpful comments on an earlier version of the manuscript.
Yuichi Yoshida is supported by JSPS KAKENHI Grant Number 24K02903.
\fi

\bibliographystyle{alpha}
\bibliography{references} 

\newcommand{\etalchar}[1]{$^{#1}$}
\begin{thebibliography}{ALO{\etalchar{+}}21b}

\bibitem[ABC{\etalchar{+}}21]{alon2021mixing}
Noga Alon, Raimundo Briceño, Nishant Chandgotia, Alexander Magazinov, and
  Yinon Spinka.
\newblock Mixing properties of colourings of the {$\mathbb{Z}^d$} lattice.
\newblock {\em Comb. Probab. Comput.}, 30(3):360--373, 2021.

\bibitem[AJK{\etalchar{+}}22]{anari2022entropic}
Nima Anari, Vishesh Jain, Frederic Koehler, Huy~Tuan Pham, and Thuy-Duong
  Vuong.
\newblock Entropic independence: Optimal mixing of down-up random walks.
\newblock In {\em STOC}, pages 1418--1430. ACM, 2022.

\bibitem[AL20]{alev2020improved}
Vedat~Levi Alev and Lap~Chi Lau.
\newblock Improved analysis of higher order random walks and applications.
\newblock In {\em STOC}, pages 1198--1211. ACM, 2020.

\bibitem[ALO20]{anari2020spectral}
Nima Anari, Kuikui Liu, and Shayan {Oveis Gharan}.
\newblock Spectral independence in high-dimensional expanders and applications
  to the hardcore model.
\newblock In {\em FOCS}, pages 1319--1330. IEEE, 2020.

\bibitem[ALO21a]{abdolazimi2022matrix}
Dorna Abdolazimi, Kuikui Liu, and Shayan {Oveis Gharan}.
\newblock A matrix trickle-down theorem on simplicial complexes and
  applications to sampling colorings.
\newblock In {\em FOCS}, pages 161--172. IEEE, 2021.

\bibitem[ALO{\etalchar{+}}21b]{anari2021logconcave}
Nima Anari, Kuikui Liu, Shayan {Oveis Gharan}, Cynthia Vinzant, and Thuy-Duong
  Vuong.
\newblock Log-concave polynomials {IV}: approximate exchange, tight mixing
  times, and near-optimal sampling of forests.
\newblock In {\em STOC}, pages 408--420. ACM, 2021.

\bibitem[ALOV19]{anari2019logconcave}
Nima Anari, Kuikui Liu, Shayan {Oveis Gharan}, and Cynthia Vinzant.
\newblock Log-concave polynomials {II}: high-dimensional walks and an {FPRAS}
  for counting bases of a matroid.
\newblock In {\em STOC}, pages 1--12. ACM, 2019.

\bibitem[ALV22]{anari2022optimal}
Nima Anari, Yang~P. Liu, and Thuy-Duong Vuong.
\newblock Optimal sublinear sampling of spanning trees and determinantal point
  processes via average-case entropic independence.
\newblock In {\em FOCS}, pages 123--134. IEEE, 2022.

\bibitem[BCC{\etalchar{+}}22]{blanca2022mixing}
Antonio Blanca, Pietro Caputo, Zongchen Chen, Daniel Parisi, Daniel
  Štefankovič, and Eric Vigoda.
\newblock On mixing of {M}arkov chains: Coupling, spectral independence, and
  entropy factorization.
\newblock In {\em SODA}, pages 3670--3692. SIAM, 2022.

\bibitem[BCLM11]{barthe2011correlation}
Franck Barthe, Dario {Cordero-Erausquin}, Michel Ledoux, and Bernard Maurey.
\newblock Correlation and {Brascamp--Lieb} inequalities for {M}arkov
  semigroups.
\newblock {\em Int. Math. Res. Not.}, 2011(10):2177--2216, 2011.

\bibitem[BCP{\etalchar{+}}21]{blanca2021entropy}
Antonio Blanca, Pietro Caputo, Daniela Parisi, Alistair Sinclair, and Eric
  Vigoda.
\newblock Entropy decay in the {Swendsen--Wang} dynamics on {$\mathbb{Z}^d$}.
\newblock In {\em STOC}, pages 1551--1564. ACM, 2021.

\bibitem[BD97]{bubley97pathcoupling}
Russ Bubley and Martin Dyer.
\newblock Path coupling: A technique for proving rapid mixing in {M}arkov
  chains.
\newblock In {\em FOCS}, pages 223--231. IEEE, 1997.

\bibitem[BDPR21]{bencs2021zero}
Ferenc Bencs, Ewan Davies, Viresh Patel, and Guus Regts.
\newblock On zero-free regions for the anti-ferromagnetic {P}otts model on
  bounded-degree graphs.
\newblock {\em Annales de l’Institut Henri Poincaré D}, 8(3):459--489, 2021.

\bibitem[BGK{\etalchar{+}}07]{bayati2007simple}
Mohsen Bayati, David Gamarnik, Dimitriy Katz, Chandra Nair, and Prasad Tetali.
\newblock Simple deterministic approximation algorithms for counting matchings.
\newblock In {\em STOC}, pages 122--127. ACM, 2007.

\bibitem[BT06]{bobkov2006modified}
Sergey~G. Bobkov and Prasad Tetali.
\newblock Modified logarithmic {S}obolev inequalities in discrete settings.
\newblock {\em Journal of Theoretical Probability}, 19(2):289--336, 2006.

\bibitem[CCC{\etalchar{+}}25]{chen2025rapidrandom}
Xiaoyu Chen, Zejia Chen, Zongchen Chen, Yitong Yin, and Xinyuan Zhang.
\newblock Rapid mixing on random regular graphs beyond uniqueness.
\newblock In {\em FOCS}, pages 2170--2193. IEEE, 2025.

\bibitem[CCE09]{carlen2009subadditivity}
Eric~A Carlen and Dario Cordero-Erausquin.
\newblock Subadditivity of the entropy and its relation to {Brascamp--Lieb}
  type inequalities.
\newblock {\em Geom. Funct. Anal.}, 19(2):373--405, 2009.

\bibitem[CCFV25]{carlson2025optimal}
Charlie Carlson, Xiaoyu Chen, Weiming Feng, and Eric Vigoda.
\newblock Optimal mixing for randomly sampling edge colorings on trees down to
  the max degree.
\newblock In {\em SODA}, pages 5418--5433. SIAM, 2025.

\bibitem[CCYZ25]{chen2025rapid}
Xiaoyu Chen, Zongchen Chen, Yitong Yin, and Xinyuan Zhang.
\newblock Rapid mixing at the uniqueness threshold.
\newblock In {\em FOCS}, pages 879--890. IEEE, 2025.

\bibitem[CDM{\etalchar{+}}19]{chen2019improved}
Sitan Chen, Michelle Delcourt, Ankur Moitra, Guillem Perarnau, and Luke Postle.
\newblock Improved bounds for randomly sampling colorings via linear
  programming.
\newblock In {\em SODA}, pages 2216--2234. SIAM, 2019.

\bibitem[CE22]{CE22}
Yuansi Chen and Ronen Eldan.
\newblock Localization schemes: A framework for proving mixing bounds for
  {M}arkov chains (extended abstract).
\newblock In {\em FOCS}, pages 110--122. IEEE, 2022.

\bibitem[Ces01]{cesi2001quasi}
Filippo Cesi.
\newblock Quasi-factorization of the entropy and logarithmic {S}obolev
  inequalities for gibbs random fields.
\newblock {\em Probab. Theory Relat.}, 120(4):569--584, 2001.

\bibitem[CFYZ22]{CFYZ22}
Xiaoyu Chen, Weiming Feng, Yitong Yin, and Xinyuan Zhang.
\newblock Optimal mixing for two-state anti-ferromagnetic spin systems.
\newblock In {\em FOCS}, pages 588--599. {IEEE}, 2022.

\bibitem[CGM21]{CGM19}
Mary Cryan, Heng Guo, and Giorgos Mousa.
\newblock Modified log-{S}obolev inequalities for strongly log-concave
  distributions.
\newblock {\em Ann. Probab.}, 49(1):506--525, 2021.

\bibitem[CG{\v{S}}V21]{chen2021rapid}
Zongchen Chen, Andreas Galanis, Daniel {\v{S}}tefankovi{\v{c}}, and Eric
  Vigoda.
\newblock Rapid mixing for colorings via spectral independence.
\newblock In {\em SODA}, pages 1548--1557. SIAM, 2021.

\bibitem[CJM{\etalchar{+}}26]{chen2026edgetilting}
Xiaoyu Chen, Zhe Ju, Tianshun Miao, Yitong Yin, and Xinyuan Zhang.
\newblock Edge-tilting field dynamics: Rapid mixing at the uniqueness threshold
  and optimal mixing for {S}wendsen-{W}ang dynamics.
\newblock arXiv/2604.10525, 2026.

\bibitem[CLL04]{carlen2004sharp}
Eric~A Carlen, Elliott~H Lieb, and Michael Loss.
\newblock A sharp analog of {Young's} inequality on {$S^N$} and related entropy
  inequalities.
\newblock {\em J. Geom. Anal.}, 14(3):487--520, 2004.

\bibitem[CLMM23]{chen2023strong}
Zongchen Chen, Kuikui Liu, Nitya Mani, and Ankur Moitra.
\newblock Strong spatial mixing for colorings on trees and its algorithmic
  applications.
\newblock In {\em FOCS}, pages 810--845. IEEE, 2023.

\bibitem[CLV20]{chen2020rapid}
Zongchen Chen, Kuikui Liu, and Eric Vigoda.
\newblock Rapid mixing of glauber dynamics up to uniqueness via contraction.
\newblock In {\em FOCS}, pages 1307--1318. IEEE, 2020.

\bibitem[CLV21]{chen2021optimal}
Zongchen Chen, Kuikui Liu, and Eric Vigoda.
\newblock Optimal mixing of {G}lauber dynamics: Entropy factorization via
  high-dimensional expansion.
\newblock In {\em STOC}, pages 1537--1550. ACM, 2021.

\bibitem[CMT15]{caputo2015approximate}
Pietro Caputo, Georg Menz, and Prasad Tetali.
\newblock Approximate tensorization of entropy at high temperature.
\newblock {\em Ann. Fac. Sci. Toulouse. Math.}, Ser. 6, 24(4):691--716, 2015.

\bibitem[CV25]{carlson2025flip}
Charlie Carlson and Eric Vigoda.
\newblock Flip dynamics for sampling colorings: Improving (11/6 —
  $\varepsilon$) using a simple metric.
\newblock In {\em SODA}, pages 2194--2212. SIAM, 2025.

\bibitem[CWZZ25]{chen2025decay}
Zejia Chen, Yulin Wang, Chihao Zhang, and Zihan Zhang.
\newblock Decay of correlation for edge colorings when {$q > 3\Delta$}.
\newblock In {\em ICALP}, volume 334, pages 54:1--54:18. Schloss Dagstuhl --
  Leibniz-Zentrum f{\"u}r Informatik, 2025.

\bibitem[DHP20]{delcourt2020glauber}
Michelle Delcourt, Marc Heinrich, and Guillem Perarnau.
\newblock The {G}lauber dynamics for edge-colorings of trees.
\newblock {\em Random Struct. Algorithms}, 57(4):1050--1076, 2020.

\bibitem[EAM22]{ElAlaoui2022information}
Ahmed El~Alaoui and Andrea Montanari.
\newblock An information-theoretic view of stochastic localization.
\newblock {\em IEEE Trans. Inf. Theory}, 68(11):7423--7426, 2022.

\bibitem[ERT80]{erdos1980choosability}
Paul Erd{\"o}s, Arthur~L Rubin, and Herbert Taylor.
\newblock Choosability in graphs.
\newblock In {\em Proc. West Coast Conf. Combinatorics, Graph Theory and
  Computing}, volume~26, pages 125--157. Utilitas Math., 1980.

\bibitem[FGYZ21]{feng2021rapid}
Weiming Feng, Heng Guo, Yitong Yin, and Chihao Zhang.
\newblock Rapid mixing from spectral independence beyond the {B}oolean
  {D}omain.
\newblock In {\em SODA}, pages 1558--1577. SIAM, 2021.

\bibitem[FS99]{ferreira1999antiferro}
Sabino~Jos{\'e} Ferreira and Alan~D. Sokal.
\newblock Antiferromagnetic potts models on the square lattice: A
  high-precision {Monte} {Carlo} study.
\newblock {\em J. Stat. Phys.}, 96(3--4):461--530, 1999.

\bibitem[GJK10]{goldberg2010mixing}
Leslie~Ann Goldberg, Mark Jerrum, and Marek Karpinski.
\newblock The mixing time of {G}lauber dynamics for coloring regular trees.
\newblock {\em Random Struct. Algorithms}, 36(4):464--476, 2010.

\bibitem[GKM15]{gamarnik2015strong}
David Gamarnik, Dmitriy Katz, and Sidhant Misra.
\newblock Strong spatial mixing of list coloring of graphs.
\newblock {\em Random Struct. Algorithms}, 46(4):599--613, 2015.

\bibitem[GMP05]{goldberg2005strong}
Leslie~Ann Goldberg, Russell Martin, and Mike Paterson.
\newblock Strong spatial mixing with fewer colors for lattice graphs.
\newblock {\em SIAM J. Comput.}, 35(2):486--517, 2005.

\bibitem[GP23]{gu2023nonlinear}
Yuzhou Gu and Yury Polyanskiy.
\newblock Non-linear log-{Sobolev} inequalities for the {Potts} semigroup and
  applications to reconstruction problems.
\newblock {\em Commun. Math. Phys.}, 404(2):769--831, 2023.

\bibitem[Hei20]{Heinrich2020glauber}
Marc Heinrich.
\newblock Glauber dynamics for colourings of chordal graphs and graphs of
  bounded treewidth.
\newblock arXiv/2010.16158, 2020.

\bibitem[HS23]{hermon2023modified}
Jonathan Hermon and Justin Salez.
\newblock {Modified {log-Sobolev} inequalities for {strong-Rayleigh} measures}.
\newblock {\em Ann. Appl. Probab.}, 33(2):1501--1514, 2023.

\bibitem[Jer95]{jerrum1995simple}
Mark Jerrum.
\newblock A very simple algorithm for estimating the number of {$k$}-colorings
  of a low-degree graph.
\newblock {\em Random Struct. Algorithms}, 7(2):157--165, 1995.

\bibitem[LLY13]{li2013correlation}
Liang Li, Pinyan Lu, and Yitong Yin.
\newblock Correlation decay up to uniqueness in spin systems.
\newblock In {\em SODA}, pages 67--84. SIAM, 2013.

\bibitem[LM11]{lucier2011glauber}
B.~Lucier and M.~Molloy.
\newblock The {G}lauber dynamics for colorings of bounded degree trees.
\newblock {\em SIAM J. Discrete Math.}, 25(2):827--853, 2011.

\bibitem[LP17]{levin2017markov}
David~A. Levin and Yuval Peres.
\newblock {\em Markov chains and mixing times}.
\newblock American Mathematical Soc., 2017.

\bibitem[LSS25]{liu2025correlation}
Jingcheng Liu, Alistair Sinclair, and Piyush Srivastava.
\newblock Correlation decay and partition function zeros: Algorithms and phase
  transitions.
\newblock {\em SIAM Journal on Computing}, 54(4):FOCS19--200--FOCS19--252,
  2025.

\bibitem[LV05]{luczak2005torpid}
Tomasz \L{}uczak and Eric Vigoda.
\newblock Torpid mixing of the {Wang–Swendsen–Koteck\'{y}} algorithm for
  sampling colorings.
\newblock {\em J. Discrete Algorithms}, 3(1):92--100, 2005.

\bibitem[LY13]{lu2013improved}
Pinyan Lu and Yitong Yin.
\newblock Improved {FPTAS} for multi-spin systems.
\newblock In {\em {RANDOM}}, pages 639--654. Springer, 2013.

\bibitem[MO94]{martinelli1994approach}
F.~Martinelli and E.~Olivieri.
\newblock Approach to equilibrium of {Glauber} dynamics in the one phase
  region. {II}. the general case.
\newblock {\em Commun. Math. Phys}, 161(3):487--514, 1994.

\bibitem[MS09]{mohar2009kempe}
Bojan Mohar and Jes{\'u}s Salas.
\newblock A new {Kempe} invariant and the (non)-ergodicity of the
  {Wang--Swendsen--Koteck{\'y}} algorithm.
\newblock {\em J. Phys. A: Math. Theor.}, 42:225204, 2009.

\bibitem[MS10]{mohar2010nonergodicity}
Bojan Mohar and Jesús Salas.
\newblock On the non-ergodicity of the {Swendsen–Wang–Kotecký} algorithm
  on the kagomé lattice.
\newblock {\em J. Stat. Mech.: Theory Exp.}, 2010(05):P05016, 2010.

\bibitem[Nar23]{narboni2023vizing}
Jonathan Narboni.
\newblock Vizing's edge-recoloring conjecture holds.
\newblock In {\em EUROCOMB}, pages 681--686. Masaryk University Press, 2023.

\bibitem[PR17]{patel2017deterministic}
Viresh Patel and Guus Regts.
\newblock Deterministic polynomial-time approximation algorithms for partition
  functions and graph polynomials.
\newblock {\em SIAM J. Comput.}, 46(6):1893--1919, 2017.

\bibitem[Sok01]{sokal2001personal}
Alan~D Sokal.
\newblock A personal list of unsolved problems concerning lattice gases and
  antiferromagnetic {Potts} models.
\newblock {\em Markov Process. Relat. Fields}, 7(1):21--38, 2001.

\bibitem[SS97]{salas1997absense}
Jesús Salas and Alan~D. Sokal.
\newblock Absence of phase transition for antiferromagnetic {Potts} models via
  the {Dobrushin} uniqueness theorem.
\newblock {\em J. Stat. Phys.}, 86(3):551--579, 1997.

\bibitem[SS22]{salas2022ergodicity}
Jes{\'u}s Salas and Alan~D. Sokal.
\newblock Ergodicity of the {Wang--Swendsen--Koteck{\'y}} algorithm on several
  classes of lattices on the torus.
\newblock {\em J. Phys. A: Math. Theor.}, 55:415004, 2022.

\bibitem[SST12]{sinclair12approximation}
Alistair Sinclair, Piyush Srivastava, and Marc Thurley.
\newblock Approximation algorithms for two-state anti-ferromagnetic spin
  systems on bounded degree graphs.
\newblock In {\em SODA}, pages 941--953. SIAM, 2012.

\bibitem[SZ92]{stroock1992logarithmic}
Daniel~W Stroock and Bogus{\l}aw Zegarli{\'n}ski.
\newblock The logarithmic {S}obolev inequality for discrete spin systems on a
  lattice.
\newblock {\em Commun. Math. Phys}, 149(1):175--193, 1992.

\bibitem[TVVY10]{tetali2010phase}
Prasad Tetali, Juan~C. Vera, Eric Vigoda, and Linji Yang.
\newblock Phase transition for the mixing time of the {G}lauber dynamics for
  coloring regular trees.
\newblock In {\em SODA}, pages 1646--1656. SIAM, 2010.

\bibitem[Vig99]{vigoda1999improved}
Eric Vigoda.
\newblock Improved bounds for sampling colorings.
\newblock In {\em FOCS}, pages 51--59. IEEE, 1999.

\bibitem[Wei06]{weitz06counting}
Dror Weitz.
\newblock Counting independent sets up to the tree threshold.
\newblock In {\em {STOC}}, pages 140--149. ACM, 2006.

\bibitem[WSK89]{wang1989antiferro}
Jian-Sheng Wang, Robert~H. Swendsen, and Roman Koteck\'y.
\newblock Antiferromagnetic {P}otts models.
\newblock {\em Phys. Rev. Lett.}, 63:109--112, 1989.

\bibitem[WSK90]{wang1990three}
Jian-Sheng Wang, Robert~H. Swendsen, and Roman Koteck\'y.
\newblock Three-state antiferromagnetic {P}otts models: A {Monte Carlo} study.
\newblock {\em Phys. Rev. B}, 42:2465--2474, 1990.

\bibitem[WZZ24]{wang2024samplingproper}
Yulin Wang, Chihao Zhang, and Zihan Zhang.
\newblock Sampling proper colorings on line graphs using {$(1+o(1))\Delta$}
  colors.
\newblock In {\em STOC}, pages 1688--1699. ACM, 2024.

\bibitem[YZ13]{yin2013approximate}
Yitong Yin and Chihao Zhang.
\newblock Approximate counting via correlation decay on planar graphs.
\newblock In {\em SODA}, pages 47--66. SIAM, 2013.

\end{thebibliography}

\clearpage

\end{document}